\newif\ifnotes
\newif\ifcr
\notesfalse
\crfalse

\newif\ifllncs
 \llncsfalse

\ifnotes
\newcommand{\omri}[1]{$\ll$\textsf{\color{red} Omri: { #1}}$\gg$}

\else
\newcommand{\omri}[1]{}
\fi

\ifllncs
\documentclass[11pt]{llncs}
\else
\documentclass[a4paper,11pt]{article}
\usepackage{amsthm}
\usepackage{fullpage}
\fi
\usepackage{footmisc} 
\usepackage[hypertexnames=false,colorlinks=true,citecolor=Maroon]{hyperref}
\usepackage{newtxtext}
\usepackage{amsmath,amssymb,enumerate,algorithm,algorithmic,latexsym}
\usepackage{tikz}
\usetikzlibrary{fpu}
\usepackage{breakcites}
\usepackage{float}
\newfloat{algorithm}{H}{lop}
\usepackage{color}
\usepackage{colortbl}
\definecolor{Maroon}{cmyk}{0, 0.87, 0.68, 0.32}
\usepackage{url}
\usepackage{graphicx}
\usepackage{xspace}
\usepackage{subfigure}
\usepackage{amsfonts}
\usepackage{caption}
\usepackage{enumitem}
\usepackage{soul}
\usepackage[normalem]{ulem}
\usepackage{url}
\usepackage{graphicx}
\usepackage{bookmark}
\numberwithin{algorithm}{section}
\usepackage{mathtools}
\usepackage{dsfont}
\usepackage{cleveref}
\renewcommand{\paragraph}[1]{\vspace{1.5mm}\noindent \textbf{#1}}

\newcommand{\Nat}{\mathbb{N}}

\newcommand{\bbC}{\mathbb{C}}

\newcommand{\lang}{\mathcal{L}}

\newcommand{\BQP}{\textbf{BQP}}
\newcommand{\PSPACE}{\textbf{PSPACE}}
\newcommand{\NP}{\textbf{NP}}

\newcommand{\SAT}{\mathsf{SAT}}

\newcommand{\Rel}{\mathcal{R}}

\newcommand{\YES}{\prod_{\text{YES}}}
\newcommand{\NO}{\prod_{\text{NO}}}

\newcommand{\CTC}{\text{CTC}}

\newcommand{\CJ}{\mathbf{CJ}}
\newcommand{\Lin}{\mathcal{L}}
\newcommand{\Hilb}{\mathcal{H}}
\newcommand{\channel}{\mathcal{C}}
\newcommand{\states}{\mathcal{S}}
\newcommand{\pmo}{\mathbf{PMO}}

\newcommand{\BQPlctc}{\BQP_{\ell\CTC}} 
\newcommand{\BQPuctc}{\BQP_{u\CTC}} 

\newcommand{\zo}{\{0, 1\}}

\newcommand{\ket}[1]{|{#1}\rangle}
\newcommand{\bra}[1]{\langle{#1}|}
\newcommand{\braket}[2]{\langle{#1}|{#2}\rangle}

\newcommand{\Dket}[1]{| {#1} \rangle\!\rangle}
\newcommand{\Dbra}[1]{\langle\!\langle {#1} |}

\newcommand{\textabbrevstyle}[1]{\mbox{#1}}
\newcommand{\textabbrevstylebol}[1]{\mbox{\textbf{#1}}}
\newcommand{\newtextabbrev}[1]{\expandafter\newcommand\csname #1\endcsname{\textabbrevstyle{#1}\xspace}}
\newcommand{\newtextabbrevbol}[1]{\expandafter\newcommand\csname #1\endcsname{\textabbrevstylebol{#1}\xspace}}
\newcommand{\renewtextabbrevbol}[1]{\expandafter\renewcommand\csname
#1\endcsname{\textabbrevstylebol{#1}\xspace}}

\newtextabbrevbol{EXP}
\newtextabbrevbol{Dtime}
\newtextabbrev{PRG}
\newtextabbrev{PRF}
\newtextabbrev{OWF}
\newtextabbrev{DPT}
\newtextabbrev{PPT}
\newtextabbrev{QPT}
\newtextabbrev{EOWF}
\newtextabbrev{GEOWF}
\newtextabbrev{ZAP}
\newtextabbrevbol{Ntime}
\newtextabbrevbol{coAM}
\newtextabbrevbol{MIP}
\newtextabbrevbol{NEXP}
\newtextabbrevbol{BPP}
\newtextabbrevbol{coNP}
\newtextabbrevbol{coMA}
\newtextabbrevbol{TFNP}
\newtextabbrev{NIWI}
\newtextabbrevbol{AM}
\newtextabbrev{SNARG}
\newtextabbrev{SNARK}
\newtextabbrev{coRP}
\newtextabbrev{NIUA}
\newtextabbrev{UA}
\newtextabbrev{ZK}
\newtextabbrev{WZK}
\newtextabbrev{WH}
\newtextabbrev{AI}
\newtextabbrev{ECRH}
\newtextabbrev{PIR}
\newtextabbrev{WIPOK}
\newtextabbrev{CDS}
\newtextabbrev{POK}
\newtextabbrev{FE}
\newtextabbrev{IO}
\newtextabbrev{XIO}
\newtextabbrev{SXIO}
\newtextabbrev{SKFE}
\newtextabbrev{PKFE}
\newtextabbrev{PPRF}
\newtextabbrev{QLWE}
\newtextabbrev{PKE}

\ifllncs
\else
\newtheorem{definition}{Definition}[section]

\newtheorem{lemma}{Lemma}[section]
\newtheorem{corollary}{Corollary}[section]
\newtheorem{theorem}{Theorem}[section]

\newtheorem{proposition}{Proposition}[section]

\theoremstyle{remark}

\newtheorem{remark}{Remark}[section]

\fi

\newcommand{\thmref}[1]{Theorem~\protect\ref{#1}}

\newcommand{\lemref}[1]{Lemma~\protect\ref{#1}}

\newenvironment{boxfig}[2]{\begin{figure}[#1]\fbox{\begin{minipage}{\linewidth}
                        \vspace{0.2em}
                        \makebox[0.025\linewidth]{}
                        \begin{minipage}{0.95\linewidth}
            {{
                        #2 }}
                        \end{minipage}
                        \vspace{0.2em}
                        \end{minipage}}}{\end{figure}}

\newcommand{\norm}[1]{\ensuremath{\left\lVert #1 \right\rVert}} %
\newcommand{\tr}{Tr} %
\renewcommand{\paragraph}[1]{\smallskip\noindent{\bf #1}}

\usepackage{bm}
\usepackage{booktabs}
\usepackage{array}
\numberwithin{equation}{section}

\newcommand{\UP}{\textbf{UP}}
\newcommand{\coUP}{\textbf{coUP}}
\newcommand{\TFUP}{\textbf{TFUP}}
\newcommand{\PP}{\textbf{PP}}

\newcommand{\USO}{\mathsf{USO}}

\newcommand{\td}{\Delta_{\mathrm{tr}}}
\newcommand{\id}{\mathrm{id}}
\newcommand{\Ad}{\operatorname{Ad}}

\hypersetup{
    pdftitle={What Can Be Computed with Indefinite Causal Structure?},
    pdfauthor={Omri Shmueli}
}

\title{What Can Be Computed with Indefinite Causal Structure?}
\ifllncs
\author{Omri Shmueli}
\institute{}
\else
\author{Omri Shmueli\thanks{omri.shmueli1@gmail.com}}

\date{}
\fi

\begin{document}
\maketitle

\begin{abstract}
Can quantum computation become more powerful when operations do not need to follow a definite causal order? We study this question in the process matrix formalism of Oreshkov et al. (Nature Communications, 2012), which allows quantum operations to be connected without a fixed global order. In the computational model, a quantum circuit succinctly describes a process matrix, and the algorithm is given access to its generated channel, with computational cost measured by the size of the generating circuit.

We give polynomial-time algorithms in two settings. First, general process matrix computation suffices to solve every problem in $\NP$. The valid processes used by our algorithm satisfy the purification postulate of Ara\'{u}jo et al. (Quantum, 2017), and their generated channels are unitary, although their generating circuits are not. Second, requiring the process matrices themselves to be pure, we give a polynomial-time algorithm for unique-sink orientation, a general problem capturing several longstanding computational problems in optimization. Our unitary generators are constructed explicitly and efficiently. Assuming that unique-sink orientation has no quantum polynomial-time algorithm, our result shows that indefinite causal order can provide computational power beyond $\BQP$ within the process matrix formalism, even while preserving unitarity.
\end{abstract}

\ifllncs
\pagestyle{plain}
\else
\newpage
\thispagestyle{empty}
\tableofcontents
\newpage
\fi
\pagenumbering{arabic}

\section{Introduction} \label{section:introduction}
One of the basic questions in our understanding of nature is whether physical events must have a definite causal structure, namely, a relation specifying which events can influence which others. In general relativity, these relations depend on the geometry of spacetime, which is itself dynamical. If spacetime geometry is subject to quantum uncertainty, as may occur in quantum gravity, it raises the possibility that its causal structure can also be indefinite \cite{hardy2005probability,hardy2007towards}. This possibility, referred to as \emph{indefinite causal structure} (ICS), motivates studying information processing without a predetermined causal order. In particular, it raises the question of whether ICS can itself be a computational resource \cite{hardy2009quantum,chiribella2013quantum}.

The quantum circuit model, based on standard quantum mechanics, allows superpositions of physical states, but the circuit still specifies a causal order (formally, a partial order) between the execution of its gates. It is therefore natural to ask whether relaxing the latter requirement can change the computational power of the model. To make this question meaningful, however, one needs to specify which connections between quantum operations are allowed, why they define consistent computations, and how to measure their computational cost.

\paragraph{Process matrices: local quantum mechanics without a fixed global order.}
Oreshkov, Costa and Brukner \cite{oreshkov2012quantum} introduced the \emph{process matrix formalism} (PMF), which has since become a central framework for studying indefinite causal order. The formalism separates two assumptions. The first is that ordinary quantum mechanics governs information processing within each local spacetime region, and the second is that globally, these regions are connected according to a definite causal order. In the PMF, these regions are modeled as laboratories, each receiving a quantum system, performing an operation, and producing an output system. A \emph{process matrix} is an algebraic object describing the correlations and possible information flow between these laboratories, without imposing a definite causal order on one hand, but also not violating local quantum mechanics.

Let us explain what are process matrices in the language of computer science. A joint probability distribution describes correlations between classical registers, while a joint density matrix describes quantum correlations between quantum registers, including entanglement. These correlations, classical or quantum, are between physical \emph{systems}. Process matrices extend this to correlations between quantum \emph{operations}. To elaborate, an ordinary quantum circuit implements a completely positive trace-preserving (CPTP) map from input registers to output registers. A joint density matrix, in contrast, can be thought of as a snapshot of registers occupying different spatial regions: it specifies their joint state, but not their evolution. This motivates viewing density matrices as describing quantum correlations between systems across space at a given time, and process matrices as generalized states describing correlations between operations across space and time.

Both density matrices and quantum channels arise naturally as special cases of process matrices. When the laboratories receive quantum systems but send none out, the process matrix reduces to a joint density matrix describing their input systems. When one laboratory's output is sent to another laboratory's input, an ordinary quantum channel describes this connection. The PMF includes both situations, while also allowing correlations between local operations that cannot be explained by a definite global causal order \cite{oreshkov2012quantum}. 

A bit more formally, consider local laboratories applying quantum channels $\mathcal A_1,\ldots,\mathcal A_k$. A process matrix $W$ describes a transformation from the local operations to a global operation $\channel_{\mathcal A_1,\ldots,\mathcal A_k}$. This transformation depends linearly on each local operation. We require $\channel_{\mathcal A_1,\ldots,\mathcal A_k}$ to be a quantum channel for \emph{every} choice of the local channels, including operations involving auxiliary registers whose inputs may be entangled across laboratories or with an external input \cite{chiribella2013quantum,araujo2017purification}. This requirement allows us to describe how local quantum operations are \emph{combined}, without either thinking of specific channels or prescribing a definite causal order between them.

\newpage

\paragraph{Purification and physicality.}
The evolution of a closed quantum system is unitary, and hence preserves information. This reversibility is a fundamental feature of quantum mechanics. General quantum channels need not be reversible, but that is only since they may describe a part of a larger system. Namely, any quantum channel can be implemented by initializing auxiliary registers, applying a unitary transformation on the larger system, and discarding some output registers -- this process is called purification. Thus, information that appears to be lost may remain in registers that are no longer observed. The PMF as was initially defined by \cite{oreshkov2012quantum} does not guarantee an analogous reversibility for transformations corresponding to process matrices.

Ara\'{u}jo, Feix, Navascu\'{e}s and Brukner \cite{araujo2017purification} later proposed a \emph{purification postulate} to impose such a requirement. Specifically, we call a process \emph{pure} if every choice of local unitary operations $\mathcal A_1,\ldots,\mathcal A_k$, including operations involving auxiliary registers, induces a global unitary transformation $U_{\mathcal A_1,\ldots,\mathcal A_k}$. A process is called \emph{purifiable} if it can be obtained from a larger pure process by fixing additional inputs to pure states and discarding additional outputs. The postulate requires physically realizable processes to be purifiable. This is indeed a nontrivial restriction: the original process of \cite{oreshkov2012quantum}, whose correlations cannot arise under any definite causal order, even allowing classical uncertainty about that order, is not purifiable \cite{araujo2017purification}.

\paragraph{From process matrices to computation.}
Indefinite causal order already offers advantages in information processing tasks. The quantum switch \cite{chiribella2013quantum}, which coherently controls the order of two operations, already provides advantages using a pure process \cite{araujo2017purification}. For example, it can determine with certainty whether two unknown unitary gates commute or anticommute using each only once, which is impossible for fixed-order circuits under the same query restriction \cite{chiribella2012perfect}. It also enables protocols that use exponentially less communication than one-way quantum protocols \cite{guerin2016exponential}, and allows two channels that individually erase all input information to transmit classical information when applied in a superposition of orders \cite{ebler2018enhanced}. These results do not settle whether pure, or even general process matrices, enable efficient computation beyond $\BQP$.

Ara\'{u}jo, Gu\'{e}rin and Baumeler \cite{araujo2017quantum} formalize process matrix computation using uniformly generated, polynomial-size ordinary quantum circuits $C$ with designated output registers connected back to input registers. In this work we call such a circuit, together with its designated connections, a \emph{process matrix generator} when the connected transformation is a quantum channel for every choice of local channels on the individual connections, including auxiliary registers as above. The computational model grants access to the resulting channel, with cost measured by the size of $C$.

This computational model was shown \cite{araujo2017quantum} to admit a description in terms of a subclass of \emph{postselection closed timelike curves} (CTCs), represented using postselected teleportation \cite{lloyd2011quantum}. To elaborate, general postselection can introduce an input-dependent normalization, making the induced evolution nonlinear. In contrast, process matrix validity imposes a strong restriction: the postselection probability must be independent of both the input state and the local operations, preserving linearity. Thus, \cite{araujo2017quantum} define the class corresponding to polynomial-time process matrix computation as $\BQPlctc$, standing for $\BQP$ with access to linear postselection CTCS. Importantly, $\BQPlctc$ corresponds to process matrix computation \emph{without} asking for the purification restriction. Pure PMGs (Definition \ref{definition:pure_process_matrix_generator}) further restrict this to \emph{unitary} CTCs: the generator is unitary, and so is the connected evolution for every choice of local unitary operations.

An efficient process matrix generator circuit gives a short description of a process, but the computation obtained by its process matrix may be much harder to implement. Postselected teleportation illustrates this difficulty: its success probability can be exponentially small, even when the connected channel is unitary. Ara\'{u}jo, Gu\'{e}rin and Baumeler \cite{araujo2017quantum} showed that $\BQPlctc$ lies between $\BQP$ and $\PP$, leaving open whether the first inclusion is strict. We explore this question, including under the further requirement that the process itself be pure: can indefinite causal structure enable computation beyond $\BQP$ while preserving the reversibility of quantum mechanics?

\subsection{Results} \label{subsection:our_results}
We investigate quantum computation with indefinite causal structure, considering both purifiable processes that are not pure, and the further restriction to pure processes. Our work has two main parts.
\begin{enumerate}
    \item
    First, we show that polynomial-time process computation can solve every problem in $\NP$. Our constructed process is furthermore purifiable (Definition \ref{definition:purifiable_process}), but not pure (Definition \ref{definition:pure_process_matrix_generator}).

    \item
    Second, we give a polynomial-time algorithm for the unique-sink orientation problem, using pure processes whose unitary generators are constructed explicitly and efficiently.
\end{enumerate}
Hardy \cite[Section 4.3]{hardy2009quantum} asked whether indefinite causal structure can enable efficient computation beyond $\BQP$. Within the process matrix framework, our results answer this question affirmatively, even when restricted to pure processes, assuming that $\USO$ has no quantum polynomial-time algorithm. These results also clarify the role of linearity in postselection closed timelike curves, and distinguish between purification as a proposed condition for physicality, and a potential condition for efficient realization.

\vspace{1mm}
As part of our work we also extend the computational framework of \cite{araujo2017quantum}. For a promise problem, the generating circuit constructed from the input may be guaranteed to define a valid process only when the promise holds. Our definition also covers inputs that do not satisfy the promise: if the circuit does not define a valid process, the algorithm is given access to an arbitrary quantum channel, and must remain correct for any such channel. On promised inputs, the computation is unchanged. We give the formal definitions in Section \ref{section:process_matrix_formalism}.

\vspace{2mm}
\paragraph{Part I: Linear CTCs, purification, and $\NP$.}
The computational power of postselection CTCs has been closely associated with their nonlinear evolution \cite[Section V]{lloyd2011quantum}, and quantum computation with unrestricted postselection CTCs has exactly the power of $\PP$ \cite{aaronson2005postselection}. The restriction to linear CTCs was motivated, in part, by the possibility that removing nonlinearity would exclude such computational power \cite[Section III]{araujo2017quantum}. We show that linearity alone does not provide this restriction.

\begin{theorem}
Let $\BQPlctc$ as in Definition \ref{definition:BQPlctc}. Then,
$$
\NP\subseteq\BQPlctc
\enspace .
$$
The containment continues to hold even given two more conditions: the valid generation queries are required to describe purifiable processes (Definition \ref{definition:purifiable_process}) whose generated channels (Definition \ref{definition:process_matrix_generator}) are unitary.
\end{theorem}

In our construction, even though the circuit generator is purifiable and the generated channel is unitary, the circuit generating the process is \emph{not} unitary, which makes the process not pure (Definition \ref{definition:pure_process_matrix_generator}). We discuss how this can be interpreted in our Technical Overview (see Section \ref{section:overview_purifiable_vs_pure}).

\vspace{2mm}
\noindent
The main construction solves unique-witness search. Given a Boolean circuit $v : \zo^n \rightarrow \zo$ with exactly one satisfying input $z \in \zo^n$, we construct in classical polynomial time a PMG whose generated channel swaps between $\ket{0}$ and $\ket{z}$ and fixes the other computational basis states. Thus, one use of the generated channel obtains $z$. To solve all of $\NP$, where the input circuit $v' : \zo^n \rightarrow \zo$ does not necessarily have exactly one witness, we use the Valiant-Vazirani reduction \cite{valiant1985np} and verify every candidate witness before accepting.

We also show how to use general process computation for coherent permutation inversion. Given a classical circuit computing a permutation $\pi:\{0,1\}^n\to\{0,1\}^n$, we efficiently construct a PMG whose generated channel implements $\ket{y}\mapsto\ket{\pi^{-1}(y)}$, including on superpositions and inputs entangled with other registers. The unique-witness and permutation-inversion constructions are given in Section \ref{section:linear_CTC_computation}.

\vspace{2mm}
\paragraph{Part II: Pure processes solve unique-sink orientation.}
We next require the generating circuits themselves to be unitary and the resulting processes to be pure and polynomial-time. We show an algorithm in this framework solving the unique-sink orientation problem ($\USO$), a well-studied problem arising in optimization \cite{szabo2001unique}. An instance specifies the directions of the edges of an $n$-dimensional Boolean cube, with the promise that every face has exactly one sink\footnote{Note that the overall number of faces is $\sum_{k \in \{ 0, 1, \ldots, n \}} (\# \text{ of $k$ dimensional faces}) = \sum_{k \in \{ 0, 1, \ldots, n \}} \binom{n}{k} 2^{n - k} = 3^n$.}. The orientation is given succinctly by a Boolean circuit that, on input a vertex, outputs the directions of its outgoing edges. The task is to find the sink of the entire cube. The most efficient known quantum algorithms for $\USO$ take $2^{\Omega(n)}$ time \cite{bacon2017quantum}. We prove the following.

\begin{theorem}[Informal; See Theorem \ref{thm:uso}]
The problem $\USO$ (Definition \ref{definition:uso}) can be solved in polynomial time using pure process computation (Definition \ref{definition:BQPuctc}). The algorithm uses the generated channel once and succeeds with probability $1$.
\end{theorem}

More specifically, our construction generate access to a coherent inversion of the USO function. Write $s(x)$ for the string whose $i$-th bit indicates whether the edge in direction $i$ leaves $x$. The USO property implies that $s$ is a permutation (Lemma \ref{lem:uso-characterization}), and our generated unitary implements
$$
T_s \cdot \ket{y}=\ket{s^{-1}(y)}
\enspace .
$$
The sink is obtained by setting $y=0^n$.

\vspace{3mm}
\noindent
\textit{Unique-sink orientation and related problems.}
The study of $\USO$ connects decades of work in optimization, game theory, and formal verification \cite{szabo2001unique,jurdzinski2008simple}. Classical polynomial-time reductions to $\USO$ cover rational P-matrix linear complementarity, simple stochastic games, discounted and mean-payoff games, parity games, one-dimensional energy games, and full modal $\mu$-calculus model checking \cite{gartner2005simple,jurdzinski2008simple,fearnley2009linear,brim2011faster,wilke2001alternating}.
These connections give the apparently simple task of finding a sink several different meanings. For example, stochastic games have been used to design controllers for autonomous transport and energy management \cite{kwiatkowska2018prism}. In cybersecurity, this includes finding optimal defensive strategies to preserve service availability against an adversary seeking to disrupt it \cite{deshpande2014stochastic}. Model checking makes it possible to verify properties of software through a formal model, even when that software's behavior is too complex to understand by inspecting its code \cite{clarke2018model}. For these problems, general polynomial-time solvability remains open in both classical and quantum computation.

\paragraph{Technical contributions: Constructing valid process matrices.}
Beyond the generating circuits themselves, our two main results introduce techniques in quantum information theory for proving process validity. For general processes, we show how an ordinary quantum circuit preparing the Choi state of an isometry can be turned into a process matrix generator implementing that isometry (Lemma \ref{lemma:from_choi_to_isometry}). This gives the construction underlying our $\NP$ result. For pure processes, our statement is slightly more general than the $\USO$ algorithm. We construct explicit polynomial-size pure process matrix generators for a family of quantum changes of basis. Each basis vector is obtained from a computational-basis state $\ket{x}$ by applying \emph{single-qubit} unitaries, each chosen between two options. The binary choices depend efficiently on $x$, and furthermore satisfy a combinatorial condition that guarantees process validity (Theorem \ref{theorem:decomposable_unitary_generation}). This combinatorial condition is exactly a characterization of unique-sink orientations, which is how the connection to $\USO$ comes up in the first place.

\paragraph{Directions for physical implementation.}
We also relate the structure of our pure generators to existing circuit simulations and quantum-switch implementation proposals (Section \ref{subsection:uso_physical_implementation}). In particular, our construction includes a product basis for which a quantum-switch measurement protocol and a photonic implementation proposal are known. Efficient physical realization for arbitrary $\USO$ instances remains open.

\subsection{Related Work} \label{subsection:related_work}

\paragraph{Process matrices and information processing.}
The process matrix formalism was introduced by Oreshkov, Costa and Brukner \cite{oreshkov2012quantum} to describe quantum correlations without assuming a definite global causal order. A sequence of works explored the information-processing advantages available in this setting. Chiribella \cite{chiribella2012perfect}, Ara\'{u}jo, Costa and Brukner \cite{araujo2014computational}, and Renner and Brukner \cite{renner2022computational} demonstrated query advantages for tasks involving unknown quantum operations. In communication, Gu\'{e}rin et al.\ \cite{guerin2016exponential} gave a protocol using exponentially less communication than any one-way quantum protocol for the same task. Other communication advantages include transmitting classical information through completely depolarizing channels and perfect quantum communication using certain channels with zero quantum capacity \cite{ebler2018enhanced,chiribella2021indefinite}. Applications also include quantum metrology and refrigeration \cite{zhao2020quantum,felce2020quantum}. More recently, Abbott, Mhalla and Pocreau \cite{abbott2024quantum,abbott2025classical} studied Boolean-function query complexity using general processes, obtaining advantages in both the error achievable with a fixed number of queries and exact query complexity. Bavaresco et al.\ \cite{bavaresco2025simulating} proved an exponential query lower bound for simulating the quantum switch: for arbitrary unknown $n$-qubit channels, a deterministic exact simulation that calls one channel only once must call the other exponentially many times. These query results concern access to unknown operations. For computation with explicitly described circuits, Ara\'{u}jo, Gu\'{e}rin and Baumeler \cite{araujo2017quantum} introduced the computational framework on which we build. They also constructed processes that violate causal inequalities but admit ordinary quantum-circuit simulations with constant overhead. Finally, Ara\'{u}jo, Feix, Navascu\'{e}s and Brukner \cite{araujo2017purification} defined pure and purifiable processes and proposed purifiability as a necessary condition for physical realizability.

\paragraph{Computational models of closed timelike curves.}
Deutsch \cite{deutsch1991quantum} introduced a model of CTCs in which the state travelling around a loop must be a fixed point of its evolution.\footnote{These are usually referred to as Deutschian CTCs in the literature.} Aaronson and Watrous \cite{aaronson2009closed} characterized the power of polynomial-time classical and quantum computation in this model as $\PSPACE$. Another model uses postselected teleportation to connect a system's output back to its input \cite{lloyd2011closed,lloyd2011quantum}.\footnote{These are usually referred to as postselection CTCs in the literature.} Results on quantum computation with postselection and its implementation using CTCs establish exactly the power of $\PP$ for unrestricted postselection CTCs \cite{aaronson2005postselection,brun2012perfect}. Both models generally induce nonlinear evolution.

\paragraph{Logical consistency and reversibility.}
Baumeler and Wolf \cite{baumeler2017noncausal,baumeler2018tameness} studied classical circuits with feedback under the requirement that, for each input, there is exactly one assignment of values to the wires consistent with the gates. They characterized the decision problems solvable in their polynomial-time non-causal circuit model as $\UP\cap\coUP$. They also obtained this class as an upper bound for classical deterministic processes, where consistency must hold under every choice of local operations. Baumeler et al.\ \cite{baumeler2019reversible} further showed that every valid deterministic classical process admits a reversible extension. They distinguished classical reversibility from validity under arbitrary local quantum operations, which requires a separate argument for a pure quantum process.

\paragraph{Experimental evidence and physical realizability.}
The physical realizability of general process matrix computation is unknown. Even for the standard quantum circuit model, constructing a scalable quantum computer remains an outstanding challenge. Nevertheless, several experiments support indefinite causal order between quantum operations. Photonic implementations of the quantum switch have been certified against classical mixtures of definite orders \cite{rubino2017experimental,goswami2018indefinite}, and its full process matrix has been reconstructed \cite{antesberger2024higher}. More recent experiments reported violations of Bell-like inequalities for causal order \cite{richter2026toward,guo2026experimental,qu2026experimental}. These conclusions rely on assumptions about the devices and their spacetime arrangement, and experimental loopholes may remain. The experiments concern the quantum switch and do not establish the physical realizability of general or even pure process matrices.

\section{Technical Overview} \label{section:technical_overview}
In this section we describe the main ideas behind our results. Before technical details, we recall how to formalize computation with process matrices.

\paragraph{How to compute with process matrices.}
In the introduction we described a process matrix as specifying quantum correlations between local quantum operations. To use such a process for efficient computation, we need a \emph{succinct} description of it. In our model, this description is an ordinary quantum circuit. For ease of presentation, we first consider unitary circuits, although the same approach applies to general quantum circuits.

Consider a unitary circuit $C$ acting on $(n+r)$ qubits. We call its first $n$ single-qubit input and output registers the \emph{feedback registers}, denoted $R_1,\ldots,R_n$, and denote the remaining registers by $Q_1,\ldots,Q_r$. For the moment, think of $C$ as a static diagram, rather than a computation being executed. Its gates and wires describe connections between the input and output registers. We will use this entire diagram to specify how local operations at the feedback registers are connected to one another and to the remaining registers. 

\paragraph{How to prove validity.}
At this point we can consider the interpretation in terms of $n$ laboratories. The circuit $C$ (or more precisely the pair $D := (n, C)$) defines the correlations between the laboratories. The PMF says that if $D$ is a PMG, then whatever operations performed in the laboratories, without causal order, the rest of the system will continue to behave like a quantum channel.

Let us explain it more carefully. Each laboratory performs an operation $\mathcal A_i$ at one of the feedback registers. $\mathcal A_i$ receives the \emph{output} of $R_i$ and returns a system to its \emph{input}. Furthermore allow each $\mathcal A_i$ to act also on an auxiliary input $E_i$ and produce an auxiliary output $E_i'$. Connect the output of each feedback register back to its input through $\mathcal A_i$, leaving $Q:=Q_1, \ldots , Q_r$ and all auxiliary registers external. The resulting transformation, denoted $\channel_{\mathcal A_1,\ldots,\mathcal A_n}$, takes an input on $Q, E_1, \ldots ,E_n$ and produces an output on $Q, E_1', \ldots , E_n'$. We require it to be a quantum channel for \emph{every} choice of local quantum channels and auxiliary registers. In particular, the inputs may be entangled across laboratories, with $Q$, or with an external reference. As said, with $D=(n,C)$ satisfies this requirement we call it \emph{process matrix generator} (PMG, Definition \ref{definition:process_matrix_generator}).

\paragraph{The generated channel.}
For the computation itself, we set all the local operations to the identity and use no auxiliary registers. Each feedback output is then connected directly to its corresponding input, leaving a transformation on $Q$. We call this the \emph{generated channel} and denote it by $\pmo_D$.

For a unitary circuit $C$, there is a simple expression for this connection (see Lemma \ref{lemma:simplified_isometry_pm} for formal statement). Writing $U_C$ for its unitary matrix and $R:=R_1, \ldots , R_n$, we have
$$
T_D:=\tr_R(U_C)
\;
,
\qquad
\pmo_D(\rho)=T_D\rho T_D^\dagger
\enspace .
$$
The partial trace here acts on the matrix $U_C$ itself: it connects the feedback input and output indices. It does \emph{not} mean executing $C$ and then discarding the feedback registers from the resulting state. For a unitary $C$, if $D$ is a valid PMG, then $T_D$ is unitary. In fact, a valid PMG with a unitary generating circuit is pure (Definition \ref{definition:pure_process_matrix_generator}): every choice of local operations that are guaranteed to be unitary, including operations on auxiliary registers, gives a global unitary transformation.

For a valid PMG $D$, the model gives us access to $\pmo_D$, with cost measured by the size of $C$. This is where the opportunity for a jump in computational power: a small generating circuit can describe a channel that may be difficult to implement by an ordinary quantum circuit. Our constructions therefore have two parts. We design a circuit whose feedback connections give the transformation we want, and then prove that the circuit remains valid for every choice of local operations.

\subsection{General process matrix computation and NP} \label{subsection:overview_np}
The main quantum-information result behind our $\NP$ algorithm is a construction that turns a circuit preparing a Choi state into a process matrix generator for the corresponding channel. We first explain this construction, and then show how to use it to solve $\NP$.

For an $n$-qubit unitary $V$, its Choi state is
$$
\ket{J_V}
:=
\frac{1}{\sqrt{2^n}}
\sum_{x\in\zo^n}\ket{x}\otimes V\ket{x}
\enspace .
$$
One can think of this state as a superposition of inputs together with their corresponding outputs.
The intuition is that a circuit $P_{V}$ that generates $\ket{J_V}$ gives a succinct description of the channel $V$, however, this description is "not accessible" -- it is unclear how to execute the action of $V$, using $P_{V}$. Here, we show that general process matrices allow this seemingly inaccessible information to be extracted. Our construction gives the following statement, presented here for unitaries.

\begin{theorem}[Informal; see Lemma \ref{lemma:from_choi_to_isometry}]
Given a quantum circuit preparing $\ket{J_V}$, there is a classical polynomial-time algorithm that constructs a PMG whose generated channel applies $V$.
\end{theorem}

\paragraph{From a Choi state to a channel.}
Let $\ket{\psi}$ an $n$-qubit state and recall what happens when we use $\ket{J_V}$ for teleportation. For each $i\in[n]$, we perform a Bell-basis measurement on the $i$-th input qubit and the $i$-th qubit of the first half of $\ket{J_V}$. Together, these measurements give a uniformly random outcome $j\in\zo^{2n}$. The remaining register (what used to be the second part of the register containing $\ket{J_{V}}$) contains $V\cdot P_j\ket{\psi}$, where $P_j$ is the Pauli operator associated with $j$ and $P_{0^{2n}}=I$. Thus, the outcome $j=0^{2n}$ gives exactly what we want, but its probability is only $2^{-2n}$.

We construct a generating circuit $C$ with $2n$ single-qubit feedback input registers and $n$ external input qubits containing the state to which we want to apply $V$. Its outputs likewise consist of $2n$ feedback qubits and $n$ external qubits. The idea is to arrange the feedback connections so that only the teleportation outcome $j=0^{2n}$ exists in the generated channel.

The circuit first measures its feedback input registers in the computational basis, obtaining $r\in\zo^{2n}$. It then prepares $\ket{J_V}$ in fresh auxiliary registers and performs the teleportation procedure on the external input and the first half of $\ket{J_V}$, obtaining $j\in\zo^{2n}$. The circuit then swaps the second half of $\ket{J_V}$ into the external registers, which now contain $V\cdot P_j\ket{\psi}$, and sets the feedback output to the computational-basis state $\ket{r\oplus j}$. The remaining registers are discarded. When each feedback output is connected directly to its input, the two strings must agree. This implies
$$
r=r\oplus j
\quad \Longleftrightarrow \quad
j=0^{2n}
\enspace .
$$
If $j$ is nonzero, at least one feedback bit is flipped, so that branch cannot contribute to the connected transformation. If $j=0^{2n}$, every value of $r$ is allowed, and the operation on the external input is $V$.

The normalization also works out exactly. Each value of $r$ contributes the channel $\rho\mapsto V\rho V^\dagger$ with weight $2^{-2n}$, and there are $2^{2n}$ possible values of $r$. Adding these weights gives the desired channel, without any input-dependent normalization.

\paragraph{Why this is a valid process.}
So far, we have connected each feedback output directly to its input, in other words, we know that \emph{if} this is a valid process matrix, it will give us the channel we want. The main technical part is proving that $C$ correctly describes a quantum correlation in the PMF, namely, that its derived operator is a valid process matrix. Formally, we need to argue that the transformation must remain a quantum channel when arbitrary local operations are inserted into the feedback slots. Our construction actually satisfies a stronger condition. We can allow a \emph{single} quantum channel to act jointly on all the feedback qubits, together with arbitrary auxiliary registers. As we explain later, this stronger validity property also guarantees that our process is purifiable. Let $\mathcal A$ be an arbitrary quantum channel acting on the $2n$ feedback qubits, possibly together with auxiliary registers.

We do not go into detail about the full proof, but in a nutshell, explicitly measuring the feedback registers inside the generating circuit $C$ helps is the key to our validity proof.\footnote{We later discuss whether the central role of measurement in this proof points to an artificiality of the model, when generating circuits are allowed to be non-unitary.} The measurement is used in the following way: different measurement outcomes contribute as separate branches, rather than interfering with one another. In a branch with teleportation outcome $j$, the external input undergoes the unitary $V\cdot P_j$, while the inserted channel $\mathcal A$ determines the operation on the auxiliary registers.\footnote{Recall that $\mathcal A$ acts on both the feedback registers and its auxiliary registers. Once the feedback connections are closed, its auxiliary registers remain part of the external input and output.} When checking trace preservation, the unitary $V\cdot P_j$ cancels with its adjoint. We are therefore left with showing that the operations induced by $\mathcal A$ on the auxiliary registers, summed over all measurement branches with their corresponding weights, preserve the trace. We show this using the fact that $\mathcal A$ is a valid quantum channel and thus trace-preserving.

To see why the sum over branches preserves the trace, write the feedback output as $a=r\oplus j$. Summing over $(r,j)$ is the same as summing over all pairs $(r,a)$. For each fixed input $a$ to the inserted channel $\mathcal A$, we are therefore summing over all its possible returned values $r$ and all its branches. Since $\mathcal A$ is trace-preserving, the corresponding Kraus operators satisfy the completeness relation: the sum of each operator's adjoint times itself is the identity on the auxiliary input. Finally, summing over the $2^{2n}$ choices of $a$ cancels the teleportation factor $2^{-2n}$. This proves validity, including on entangled inputs. Lemma \ref{lemma:from_choi_to_isometry} gives the full calculation.

\paragraph{Using the construction to solve $\NP$.}
We first recall that the Valiant--Vazirani reduction \cite{valiant1985np} transforms a satisfiable Boolean circuit into a different formula with exactly one satisfying input, with inverse-polynomial probability. It also preserves non-satisfiability of formulas (see Theorem \ref{theorem:valiant_vazirani}). This suggests our approach to solving $\NP$, by solving $\SAT$: assume momentarily we could take care of unique-witness finding. Then, given a $\SAT$ formula, repeat the VV reduction polynomially many times independently, use our process matrix algorithm to search for a witness in each resulting instance, and verify every returned candidate witness against the VV-reduced formula. If our unique-witness finder works, then whenever \emph{any} reduced instance has a unique witness, our construction will find it. On the remaining instances, whatever channel is returned, an incorrect candidate cannot pass the final classical check.

We therefore focus on unique-witness search. We are given a Boolean circuit $v:\zo^n\rightarrow\zo$ with exactly one satisfying input $z$, and would like to find $z$. For simplicity, assume that $z\neq 0^n$; we can ensure this by adding one input bit and requiring it to be $1$ for verification to pass. Next, consider the unitary $S_z$ that exchanges $\ket{0^n}$ and $\ket{z}$, and fixes every other computational-basis state. If we had access to $S_z$, finding the witness would be immediate by applying it to $\ket{0^n}$. However, the direct way to construct a circuit for $S_z$ uses the very string $z$ we are trying to find. The theorem above gives us another approach: prepare its Choi state using only the verifier $v$, and turn this preparation into a PMG whose generated channel applies $S_z$.

\paragraph{Generating the Choi state of $S_z$.}
For $S_z$, the Choi state is
$$
\ket{J_{S_z}}
:=
\frac{1}{\sqrt{2^n}}
\left(
\sum_{x\notin\{0^n,z\}}\ket{x,x}
+
\ket{0^n,z}
+
\ket{z,0^n}
\right)
\enspace .
$$
Almost all terms have the simple form $\ket{x,x}$. Only two terms are different, and we can recognize when to produce them by checking $v(x)$.

Start with a uniform superposition over pairs $(x,b)\in\zo^n\times\zo$. If $x\neq 0^n$, $v(x)=0$, and $b=0$, write $(x,x)$ into an output register. If $v(x)=1$, write $(0^n,x)$ when $b=0$, and $(x,0^n)$ when $b=1$. Conceptually, we mark these branches as successful. There are overall $2^{n + 1}$ branches and $2^n$ successful branches: one for each $x\notin\{0^n,z\}$, and two for $x=z$. They produce exactly the terms of $\ket{J_{S_z}}$, with equal amplitudes.

\begin{center}
\begin{tabular}{c|c|c|c}
Condition on $x$ & $b$ & First output & Second output \\
\hline
$x\neq 0^n$ and $v(x)=0$ & $0$ & $x$ & $x$ \\
$v(x)=1$ & $0$ & $0^n$ & $x$ \\
$v(x)=1$ & $1$ & $x$ & $0^n$
\end{tabular}
\end{center}

Observe that, on the successful branches, we have the following normalized state:
$$
\frac{1}{\sqrt{2^n}}
\left(
\sum_{x\notin\{0^n,z\}}\ket{x,0}\otimes\ket{x,x}
+
\ket{z,0}\otimes\ket{0^n,z}
+
\ket{z,1}\otimes\ket{z,0^n}
\right)
\enspace .
$$
In each term, the first ket contains the original registers $(x,b)$, and the second contains the two output registers. The outputs contain the pairs we want, but they are still entangled with $(x,b)$. We therefore need to uncompute these original registers. The point is that on successful branches this is easy: the output tells us what $(x,b)$ was. If its first component is zero, then $b=0$ and its second component is $x$; if its second component is zero, then $b=1$ and its first component is $x$; otherwise, $b=0$ and both components equal $x$. We can therefore uncompute $(x,b)$ and the work registers, leaving the outputs in the state $\ket{J_{S_z}}$ on the successful branches. We showed how to uncompute information only once we are at the state of successful branches. However, the successful branches take exactly $1/2$, independently of $z$, and as we show in the body of the paper a constant number of amplitude-amplification operations makes the preparation exact iwth probability $1$; see Section \ref{subsection:unique_witness_search}.

Notice that preparing this state does not by itself give an efficient way to find $z$. Measuring it in the computational basis reveals $z$ only in the two exceptional terms, whose total probability is $2/2^n$. However, our general construction turns this state preparation into a PMG implementing $S_z$. This gives the unique-witness solver needed for the Valiant--Vazirani reduction, completing the algorithm for $\NP$.

\paragraph{Inverting a permutation.}
The same Choi state transformation from Lemma \ref{lemma:from_choi_to_isometry} also gives coherent permutation inversion. Suppose we are given a classical circuit computing a permutation $\pi:\zo^n\rightarrow\zo^n$, then we can efficiently prepare
$$
\frac{1}{\sqrt{2^n}}
\sum_{x\in\zo^n}\ket{\pi(x), x}
=
\frac{1}{\sqrt{2^n}}
\sum_{ y\in \zo^n } \ket{y, \pi^{-1}(y)}
=
\ket{J_{\pi^{-1}}}
$$
simply by evaluating $\pi$ in superposition followed by swapping the input and output registers. The construction above turns this state preparation into a PMG whose generated channel applies $\pi^{-1}$, including on superpositions and inputs entangled with other registers (see Section \ref{sec:permutations}).

\subsection{Purifiable versus pure processes} \label{section:overview_purifiable_vs_pure}
Every process obtained by our Choi-state construction from Section \ref{subsection:overview_np} is actually purifiable. Note that it is purifiable (Definition \ref{definition:purifiable_process}) but not pure (Definition \ref{definition:pure_process_matrix_generator}), because its generating circuit is not unitary. To establish purifiability, we need to show that the process has a pure extension, that is, a larger unitary transformation that remains a valid process under every choice of local operations, and that recovers the original process when we initialize its additional inputs and discard its additional outputs. As we show in Proposition \ref{proposition:joint_validity_implies_purifiability}, the reason for its purifiability is that validity holds even when a single channel $\mathcal A$ acts \emph{jointly} on all $2n$ feedback qubits -- this is a substantially stronger guarantee than validity holding against disjoint spacetime regions. We can therefore regard the entire process as a transformation that takes one quantum channel $\mathcal A$ and produces another, $\channel_{\mathcal A}$. We combine this fact with Theorem 1 of \cite{chiribella2008transforming} to establish the existence of a pure extension.

\paragraph{Is there an efficiency gap between purifiable and pure processes?}
There are two details to observe in what we've seen in this section so far. The first is that in the Choi-state transformation (Lemma \ref{lemma:from_choi_to_isometry}) and its derived constructions for $\NP$ and permutation inversion, measurements are instrumental to our proof. The measurements make the proof work, but in our construction they also make the process non-pure. The second thing to observe is that the purification argument above does not tell us how to find the purification \emph{efficiently}. This suggests a possible gap that does not exist in ordinary quantum computation. Specifically, finding a purification of a given quantum circuit is always easy, and can be done by a classical deterministic polynomial-time algorithm. Here, we do not know whether that is the case. Concretely: given a purifiable PMG that admits a polynomial-size pure PMG extension, is there a polynomial-time procedure, classical or quantum, that finds such an extension? Note, for example, that a straightforward purification procedure applied to the generating circuit from Lemma \ref{lemma:from_choi_to_isometry} gives a unitary quantum circuit, but the resulting circuit is no a PMG anymore: on some inputs to the additional registers introduced by purification, the feedback connections induce a transformation that is not trace-preserving. 

A different question concerns physical realizability, which we are less equipped to answer. Here we are asking about efficient physical realization, rather than existence alone: in a theory that assumes pure processes are physical and have computational cost proportional to the generating circuit's size, is the same true for purifiable processes that are not pure? Concretely, in such a theory, assume we have a purifiable but non-pure PMG. Can one efficiently obtain access to its generated channel without first finding its purification? This question is currently difficult to answer, as we do not even know whether all pure processes are physical.

\vspace{1mm}
\begin{remark}[An example of a pure extension that uses the witness]
    Consider our unique-witness construction. We only need the Boolean verifier $v$ for it and the witness $z$ is of course unknown. From $v$, we efficiently construct a circuit $P_V$ that prepares $\ket{J_{S_z}}$, and use $P_V$ inside our PMG $C$. This circuit prepares the Choi state, performs the teleportation measurements, and sets the feedback outputs according to their outcomes. None of these steps requires knowing $z$, yet the generated channel (after tracing out the feedback registers) applies $S_z$ and therefore computes $z$ on input $\ket{0^n}$.

    Proposition \ref{proposition:joint_validity_implies_purifiability} guarantees that this PMG has a pure extension. So, what happens if we apply this argument to our concrete PMG arising solving unique witness search? For this particular construction, the extension can be chosen to have a polynomial-size generating circuit that uses $S_z$ once, together with elementary unitary operations. We leave verification of this construction to the reader. However, to write down the circuit implementing $S_z$, our construction uses $z$ explicitly: its gates implement the exchange of $\ket{0^n}$ and $\ket{z}$. Thus, given $z$, we can efficiently construct and execute the unitary generating circuit for the pure extension. But the whole point of the model is to find $z$ given $v$. The original generator is efficiently constructed from $v$ alone, whereas constructing this particular description of its pure extension first requires the very witness we are trying to find.

    This proves that a polynomial-size pure extension exists for every unique-witness instance. It does not give an efficient procedure for finding that extension from $v$, and whether a different construction can do so remains open. The purification postulate requires existence; it does not impose this computational requirement. In our next result, we construct the pure generators themselves explicitly and efficiently from the input.
\end{remark}

\subsection{Pure processes and unique-sink orientation}
\label{subsection:overview_uso}
We now construct the unitary generating circuit directly from the input, and prove that it describes a pure process. Our goal is to solve unique-sink orientation ($\USO$), defined formally in \ref{definition:uso}. We will approach this through permutation inversion again, however the technical chronology here will be different. We'll observe that from a relatively simple idea, the description of $\USO$ arises naturally. As it turns out, the condition that makes our validity proof work is exactly a characterization of unique-sink orientations.

\paragraph{Idea: try to invert a permutation unitarily.}
Suppose we are given a classical circuit computing a permutation $\pi:\zo^n\rightarrow\zo^n$, and would like to apply its inverse. Here's an idea for inverting it -- first just thinking of the generated channel, without thinking whether circuit is a valid PMG. Define
$$
h\left( x \right):=x\oplus\pi\left( x \right)
\enspace .
$$
We construct a unitary circuit $C$ on $2n$ qubits, with $n$ single-qubit feedback registers $R_1,\ldots,R_n$, collectively denoted by $R$, and an $n$-qubit external register $Q$.

Let us first describe the ordinary computation of $C$. On input $\ket{x}_R\otimes\ket{y}_Q$, compute $h\left( x \right)$ reversibly into fresh work registers, and add it bitwise modulo $2$ to $Q$. The external register now contains $y\oplus h\left( x \right)$. Since $R$ still contains $x$, we can reverse the computation of $h$ and return the work registers to zero. Finally, swap $R$ and $Q$. The resulting unitary acts as
$$
U_C\left( \ket{x}_R\otimes\ket{y}_Q \right)
=
\ket{y\oplus h\left( x \right)}_R\otimes\ket{x}_Q
\enspace .
$$
This circuit is efficiently constructed from the circuit for $\pi$. In particular, all we used was the ability to evaluate $\pi$ in the forward direction.

Now connect each feedback output directly to its input. Recall that this connection takes the partial trace of the matrix $U_C$ (see our simplification Lemma \ref{lemma:simplified_isometry_pm}). A term contributes only when the feedback input $x$ equals the feedback output $y\oplus h\left( x \right)$. By our choice of $h$,
$$
x=y\oplus h\left( x \right)
\quad \Longleftrightarrow \quad
y=\pi\left( x \right)
\enspace .
$$
For each external input $y$, exactly one string $x$ satisfies this equation, namely $x=\pi^{-1}\left( y \right)$. This is also the string placed in the external output. Therefore,
$$
\tr_R\left( U_C \right)\ket{y}
=
\ket{\pi^{-1}\left( y \right)}
\enspace .
$$
By linearity, the same transformation inverts $\pi$ coherently on superpositions.

So far, this looks like a way to invert any efficiently computable permutation using a pure process. But we have not yet proved that our circuit describes a valid quantum correlation between operations, namely, whether it is a PMG. We only know that the identity feedback connections give a unitary transformation. We still need to show that arbitrary local operations can be inserted into these connections without making the resulting transformation invalid. This is where we will need a restriction on $\pi$.

\paragraph{What makes the circuit describe a valid quantum correlation?}
We begin with local unitaries that have no auxiliary registers. Let $V_i$ be an arbitrary single-qubit unitary inserted into feedback pair $i$, and denote the feedback-connected transformation by $T_{\mathbf V}$, that is:
$$
T_{\mathbf V}
:=
\tr_R
\left[
\left(
\left( \bigotimes_{i=1}^{n}V_i \right)\otimes I_{\Hilb_Q}
\right)
U_C
\right]
\enspace .
$$
Our goal is to show that $T_{\mathbf V}$ is unitary for every choice of the $V_i$.

Fix a feedback input string $x$. Before the swap, the circuit applies $X^{h_i\left( x \right)}$ to qubit $i$ of $Q$: it either leaves that qubit alone or flips it. After the swap, that qubit goes through $V_i$ and must return to feedback input $x_i$. Its contribution is therefore the following row in the $2 \times 2$ matrix $V_i \cdot X^{h_i\left( x \right)}$,
$$
\bra{x_i}V_i \cdot X^{h_i\left( x \right)}
\enspace .
$$
Since the circuit places $x$ in the external output, the entire row of $T_{\mathbf V}$ corresponding to output $x$ is
$$
\bra{x}T_{\mathbf V}
=
\bigotimes_{i=1}^{n}
\bra{x_i}V_i \cdot X^{h_i\left( x \right)}
\enspace .
$$
Each factor is a row of a unitary matrix, so every full row has norm one. What remains is to show that different rows are orthogonal.

Take two strings $x\neq x'$. Because their rows are tensor products, it suffices to find one coordinate where the corresponding factors are orthogonal. Suppose that at some coordinate $i$,
$$
x_i\neq x_i'
\quad \text{and} \quad
h_i\left( x \right) = h_i\left( x' \right)
\enspace .
$$
At this coordinate, the circuit applies the same bit flip for both strings. We are therefore looking at different rows of the \emph{same} unitary matrix $V_i \cdot X^{h_i\left( x \right)}$. These small rows are thus orthogonal, regardless of the choice of $V_i$. Consequently, the two full rows, which are in turn obtained by tensoring the small rows, are orthogonal as well.

This gives us the following sufficient condition:
$$
x\neq x'
\quad \Longrightarrow \quad
\exists i\in[n]:
\quad
x_i\neq x_i'
\ \text{ and }\
h_i\left( x \right)=h_i\left( x' \right)
\enspace .
$$
In words, for every distinct strings, look at the coordinates where they differ -- on one of these coordinates the circuit $h(\cdot)$ has to output the same bit on both inputs. That coordinate preserves orthogonality even after an arbitrary local unitary is inserted.

\paragraph{Where unique-sink orientation enters.}
Guided by the above condition on permutations, we look to the literature, and find that it has a combinatorial meaning. Recall that a unique-sink orientation \cite{szabo2001unique} directs the edges of the Boolean cube so that every face, including the entire cube, has exactly one sink. Its outmap $s:\zo^n\rightarrow\zo^n$ records the outgoing edges: $s_i\left( x \right)=1$ if the edge in direction $i$ is going from $x$ to the outside\footnote{Formally, from $x$ to $x \oplus e_{i}$, where $e_{i} \in \zo^n$ is the vector with all zeros except a single $1$ in coordinate $i$.}, and $s_i\left( x \right)=0$ otherwise. The input to our algorithm is a Boolean circuit computing $s$, and the task is to find the vertex $z$ with $s\left( z \right)=0^n$.

A characterization of USOs states that an outmap $s$ describes a unique-sink orientation if and only if
$$
x\neq x'
\quad\Longrightarrow\quad
\exists i\in[n]:
\quad
x_i\neq x_i'
\ \text{ and }\
s_i\left( x \right)\neq s_i\left( x' \right)
\enspace ;
$$
see \cite{szabo2001unique} and Lemma our \ref{lem:uso-characterization}. In particular, this implies that $s$ is a permutation: distinct inputs always have distinct outputs.

We can therefore set $\pi=s$ in our construction. To see why this gives the condition needed for validity, recall that $h\left( x \right)=x\oplus s\left( x \right)$. At the coordinate supplied by the USO characterization, both $x_i$ and $s_i\left( x \right)$ change when we move from $x$ to $x'$, so their XOR stays the same. Thus,
$$
h_i\left( x \right)=h_i\left( x' \right)
\enspace .
$$
Conversely, when $x_i\neq x_i'$, equality of these XORs implies $s_i\left( x \right)\neq s_i\left( x' \right)$. The USO characterization is therefore exactly the condition we obtained from the unitarity calculation.

\paragraph{Including auxiliary registers.}
To finish proving purity, we must allow each $V_i$ to act jointly on its feedback qubit \emph{and} an arbitrary auxiliary register $E_i$. The inputs to these registers may also be entangled. We explain why the same orthogonality argument continues to work, by generalizing vectors to operators.

For a bit $a$, write
$$
B_{i,a}:=
V_i\left( X^a\otimes I_{\Hilb_{E_i}} \right)
\enspace .
$$
Previously, fixing the input and output bits of the feedback qubit gave a scalar. Now it gives an operator on $E_i$. Denote this operator by $B_{i,a}^{u,v}$, where $v$ is the feedback input bit and $u$ is the feedback output bit. Unitarity of $B_{i,a}$ gives
$$
\sum_{v\in\zo}
B_{i,a}^{u,v}
\left( B_{i,a}^{u',v} \right)^\dagger
=
\begin{cases}
I_{\Hilb_{E_i}}, & u=u' \;,\\
0, & u\neq u' \;.
\end{cases}
$$
This is the same statement about rows as before: each row has unit norm and different rows are orthogonal, except that the entries are now operators.

We use this identity to check that $T_{\mathbf V}\cdot T_{\mathbf V}^\dagger=I_{\Hilb_{Q,E}}$ by a block-wise analysis. Fix two external output strings $x,x'$ for the $Q$ register, and consider the corresponding block of $T_{\mathbf V}\cdot T_{\mathbf V}^\dagger$, which is an operator on $E$. If $x=x'$, every coordinate contributes the identity on $E_i$. Otherwise, the USO condition gives a coordinate $i$ with $x_i\neq x_i'$ and $h_i\left( x \right)=h_i\left( x' \right)$. At that coordinate we multiply different block rows of the same unitary $B_{i,h_i\left( x \right)}$, obtaining zero. The entire off-diagonal block is therefore zero. This proves the unitarity condition from Definition \ref{definition:pure_process_matrix_generator}, and hence our circuit is a pure PMG.

This concludes the pure construction for $\USO$: given the circuit for a USO outmap $s$, construct $C$ as above and use its generated channel, whose unitary acts as
$$
T_s\ket{y}
=
\ket{s^{-1}\left( y \right)}
\enspace .
$$
On input $\ket{0^n}$, this gives the unique sink with probability $1$. The generating circuit is constructed in classical polynomial time from the circuit for $s$, and the algorithm uses its generated channel once; see Theorem \ref{thm:uso}.

\paragraph{Beyond permutation inversion.}
Looking back at the proof, the bit flips were not essential. At each coordinate, what mattered was that equal values of $h_i$ selected the same unitary. We can therefore replace the choices $I,X$ by any two single-qubit unitaries $U_{i,0},U_{i,1}$, given by unitary circuits. The generating circuit computes $h$, applies the selected gates to $Q$ (that is, we execute $U_{i, h_{i}(x)}$ on register $Q_{i}$), uncomputes $h$, and swaps the registers, just as before. The same validity argument applies.

More explicitly, define
$$
U\left( x \right)
:=
\bigotimes_{i=1}^{n}U_{i,h_i\left( x \right)}
\qquad\text{and}\qquad
\ket{b_x}:=U\left( x \right)^\dagger\ket{x}
\enspace .
$$
The condition on $h$ guarantees that the vectors $\ket{b_x}$ form an orthonormal basis, and the generated unitary satisfies
$$
T\ket{b_x}=\ket{x}
\enspace .
$$
This describes our overall general result, given in Theorem \ref{theorem:decomposable_unitary_generation}. The permutation construction is the case $U_{i,0}=I$, $U_{i,1}=X$, where $\ket{b_x}=\ket{s\left( x \right)}$.

\section{Preliminaries} \label{section:preliminaries}
We use the following standard notions and notations from classical and quantum computation in this work.
\begin{itemize}
    \item For $n\in\Nat$, define $[n]:=\{1,2,\cdots,n\}$. The set of all finite binary strings is $\{0,1\}^{*}:=\bigcup_{n\geq0}\{0,1\}^{n}$.
    
    \item Unless stated otherwise, $\log(\cdot)$ is the base-$2$ logarithm, $\norm{\cdot}$ is the Euclidean norm of a vector, and $\oplus$ is bitwise addition modulo $2$ (aka, bitwise XOR).
    
    \item For $n\geq0$, the $n$-qubit state space is $\states(n):=\bbC^{2^n}$. In particular, it is known that $\states\left( n \right)$ is a $2^{n}$-dimensional Hilbert space. A register is a copy of such a space with its specified computational basis. Registers of the same dimension may represent different systems.
    
    \item A pure quantum state is a unit vector $\ket{\psi}$ in $\states(n)$. A mixed quantum state is a matrix $\rho\in\bbC^{2^n\times2^n}$ which is positive semidefinite and satisfies $\tr(\rho)=1$. Positive semidefinite means that $\bra{v}\rho\ket{v}\geq0$ for every vector $\ket{v}$.
    
    \item 
    For $n, m \in \Nat$, an $n$ to $m$ (general) quantum circuit is a quantum circuit $C$ with input an $n$-qubit mixed state, can perform any quantum gate from the fixed set of single-qubit gates $\{H,T,T^{\dagger}\}$, the two-qubit quantum gates $\{ CNOT, SWAP \}$ and three-qubit Toffoli gate $Tof$, execute measurement gates and add ancillary qubit registers containing $\ket{0}$, where each of the above operations costs $1$ time unit. At the end of the execution, $C$ outputs an $m$-qubit mixed state.

    \item 
    For $n \in \Nat$, an $n$-qubit unitary circuit is an $n$ to $n$ quantum circuit such that there exists a unitary matrix $U_{C} \in \bbC^{ 2^{n} \times 2^{n} }$ such that for every $n$-qubit pure state $\ket{\psi}$ the action of $C$ corresponds to the action of the unitary, that is, $C\left( \ket{\psi}\bra{\psi} \right) = U_{C} \cdot \ket{\psi}\bra{\psi} \cdot U_{C}^{\dagger}$.
\end{itemize}

\paragraph{Circuit conventions.}
The single-qubit gates in the fixed set are
$$
H:=\frac{1}{\sqrt{2}}\begin{pmatrix}1&1\\1&-1\end{pmatrix},
\qquad
T:=\begin{pmatrix}1&0\\0&e^{i\pi/4}\end{pmatrix},
$$
and $T^{\dagger}$ is the adjoint of $T$. Measurements are in the computational basis. General quantum circuits may also discard registers and use measurement outcomes for classical control. A unitary circuit may use additional work registers, provided these registers are returned to $\ket{0}$ on every input to the stated unitary. The circuit description includes its register sizes and gate count, type and position, which sums up to $|C|$. Input and output registers with zero qubits are allowed, with $\states(0)=\bbC$ and empty tensor products equal to the one-dimensional identity. The definitions below use the same conventions for zero-qubit registers.

\subsection{Computational Complexity Theory}

\begin{definition}[Promise Problem and Language] \label{definition:language_and_promise_problem}
A promise problem $\prod:=\left( \YES, \NO \right)$ is a pair of disjoint sets $\YES,\NO\subseteq\{0,1\}^{*}$. It is a decision problem when $\YES\cup\NO=\{0,1\}^{*}$, and in this case its language is $\lang:=\YES$.
\end{definition}

\begin{definition}[Search Problem]
A search problem is specified by a relation $\Rel\subseteq\{0,1\}^{*}\times\{0,1\}^{*}$ and a promised input set $\mathcal{P}$. An algorithm solves the problem with probability $p$ if, on every $x\in\mathcal{P}$, it outputs a string $z$ satisfying $(x,z)\in\Rel$ with probability at least $p$. All relations used below have polynomially bounded output length.
\end{definition}

\begin{definition}[The Complexity Class $\BQP$ and Promise $\BQP$]
The class $\BQP$ consists of languages decidable by a uniform polynomial-time quantum algorithm with probability at least $1 - 2^{-|x|}$ on every input. The analogous promise class requires correctness only on $\YES\cup\NO$.
\end{definition}

\begin{definition}[The Complexity Class $\UP$]
A promise problem $\prod:=\left( \YES, \NO \right)$ is in $\UP$ if there is a deterministic polynomial-time verifier $V(x,z)$ and a polynomial $p$ such that, for every $x \in \YES$ there is a single $z \in \{0,1\}^{p(|x|)}$ such that $V(x,z)=1$ and for $x \in \NO$ there is no $z \in \{ 0, 1 \}^{*}$ such that $V(x,z)=1$.
\end{definition}

\begin{definition}[The Complexity Class $\coUP$]
A promise problem $\prod:=\left( \YES, \NO \right)$ is in $\coUP$ if there is a deterministic polynomial-time verifier $V(x,z)$ and a polynomial $p$ such that, for every $x \in \YES$ it follows that for every $z \in \{0,1\}^{p(|x|)}$ we have $V(x,z) = 1$ and for $x \in \NO$ there is exactly one $z \in \{ 0, 1 \}^{p(|x|)}$ such that $V(x,z)=0$.
\end{definition}

\begin{definition}[The Complexity Class $\TFUP$]
$\TFUP$ is the class of search problems where the relations are polynomial-time verifiable, has polynomially bounded witnesses, and has exactly one witness on every input.
\end{definition}

\paragraph{The Satisfiability Problem $\SAT$.}
We define the satisfiability problem $\SAT$, which is known to be $\NP$-complete \cite{karp2010reducibility}. Recall that $\SAT$ is a decision problem and in the below definition, the input is an arbitrary binary string in $\{ 0, 1 \}^{*}$, which is then interpreted as a CNF formula $\varphi$. In case the input is not in the correct format to be interpreted as a CNF formula (which can be checked by a classical efficient algorithm) the input is counted as a $\NO$ instance.

\begin{definition} [The Satisfiability Problem] \label{definition:SAT}
    The Satisfiability problem, denoted $\SAT$ is a decision problem, where the input is a CNF formula $\varphi$ with $n$ input Boolean variables. The problem is defined as follows.
    \begin{itemize}
        \item
        $\YES$ is the set of all (binary encodings of) CNF formulas that are satisfiable.

        \item
        $\NO$ is $\{ 0, 1 \}^{*} \setminus \YES$.
    \end{itemize}
\end{definition}

\paragraph{The Valiant-Vazirani Theorem.}
We use the Valiant-Vazirani Theorem \cite{valiant1985np}, which can be stated as follows.

\begin{theorem} [Valiant-Vazirani Theorem] \label{theorem:valiant_vazirani}
    There exists a classical probabilistic polynomial-time Turing machine $T_{VV}$ that gets as input a CNF formula $\varphi$ that has input $n$ Boolean variables, and outputs a CNF formula $\varphi^{*}$ that has polynomially many (in $n$) Boolean variables, with the following guarantees.
    \begin{itemize}
        \item
        If $\varphi \notin \SAT$ then $\Pr_{ \varphi^{*} \gets T_{VV}\left( \varphi \right) }\left[ 
\varphi^{*} \notin \SAT \right] = 1$.

        \item
        If $\varphi \in \SAT$ then $\Pr_{ \varphi^{*} \gets T_{VV}\left( \varphi \right) }\left[ 
\varphi^{*} \text{ has exactly one satisfying assignment } \right] \geq \Omega\left( \frac{1}{n} \right)$.
    \end{itemize}
\end{theorem}

\subsection{Quantum Information Theory}
We define linear operators over Hilbert spaces. In the scope of this work we will always define Hilbert spaces as complex vector spaces, for different amounts of qubits\footnote{More precisely, we think of Hilbert spaces (with the notation $\Hilb_{A}$) as a variable that takes a value $\states\left( n \right)$ for some $n \in \Nat$. This in particular means that we can have two different Hilbert spaces that capture the same spaces, but describe the states of non-identical physical quantum systems.}.

\begin{definition} [Set of linear operators between two vector spaces]
    Let $A$, $B$ two complex vector spaces. The set $\Lin\left( A, B \right)$ is defined as the set of all linear transformations from $A$ to $B$.
    \begin{itemize}
        \item 
        In case $A = B$, we write $\Lin\left( A \right)$ for short.
    
        \item
        The identity map on $A$ is denoted by $I_{ \Lin\left( A \right) }$.

        \item 
        In case the sets are Hilbert spaces (and in our case, state spaces), then the set of linear operators is the set of appropriate matrices. Formally, if there are $n, m \in \Nat$ such that $\Hilb_{0} = \states\left( n \right)$, $\Hilb_{1} = \states\left( m \right)$, then $\Lin\left( \Hilb_{0}, \Hilb_{1} \right) = \bbC^{ 2^{m} \times 2^{n} }$, and the action of a linear transformation is vector-matrix multiplication, where the matrix is the left multiplicand. In that case we denote $I_{ \Lin\left( A \right) } := I_{ \Hilb_{0} } := I_{n}$ and $I_{ \Lin\left( B \right) } := I_{ \Hilb_{1} } := I_{m}$.
    \end{itemize}
\end{definition}

\begin{definition} [Unitary matrix and isometry] \label{definition:isometry}
    Let $\Hilb_0=\states(n)$ and $\Hilb_1=\states(m)$. A matrix $U\in\Lin(\Hilb_0)$ is unitary if $U^{\dagger}U=UU^{\dagger}=I_{\Lin(\Hilb_0)}$. A matrix $V\in\Lin(\Hilb_0,\Hilb_1)$ is an isometry if $V^{\dagger}V=I_{\Lin(\Hilb_0)}$; here the input and output spaces need not have the same dimension. We write $\Ad_V$ for the linear map $\Ad_V(X):=VXV^{\dagger}$.
\end{definition}

An operator which we will use extensively in this work is taking the partial trace of a square matrix.

\begin{definition} [Partial trace of an operator] \label{definition:partial_trace}
    Let $n, m \in \Nat$ and a linear operator $M \in \Lin\left( \Hilb_{0} \otimes \Hilb_{1} \right)$ for $\Hilb_{0} = \states\left( n \right)$, $\Hilb_{1} = \states\left( m \right)$, that is, $M \in \bbC^{ 2^{m + n} \times 2^{m + n} }$ is a square matrix. The partial trace of $M$ with respect to $\Hilb_{0}$ is defined as the following matrix in $\bbC^{2^m \times 2^m}$:
    $$
    \tr_{\Hilb_{0}}\left( M \right)
    :=
    \sum_{ x \in \{ 0, 1 \}^{n} }
    \left( \bra{x} \otimes I_{\Lin\left( \Hilb_{1} \right)} \right)
    \cdot M \cdot
    \left( \ket{x} \otimes I_{\Lin\left( \Hilb_{1} \right)} \right)
    \enspace .
    $$
    We can also write $\tr_{n}\left( M \right) := \tr_{\Hilb_{0}}\left( M \right)$.
\end{definition}

\begin{definition} [Partial trace with different external input and output spaces] \label{definition:rectangular_partial_trace}
    Let $\Hilb_R=\states(n)$, $\Hilb_P=\states(r)$, $\Hilb_F=\states(\ell)$, and let $M\in\Lin(\Hilb_R\otimes\Hilb_P,\Hilb_R\otimes\Hilb_F)$. The partial trace over the two copies of $\Hilb_R$ is the matrix
    $$
    \tr_R(M)
    :=
    \sum_{x\in\{0,1\}^{n}}
    \left(\bra{x}\otimes I_{\Lin(\Hilb_F)}\right)
    M
    \left(\ket{x}\otimes I_{\Lin(\Hilb_P)}\right)
    \enspace .
    $$
    This extends Definition \ref{definition:partial_trace}; when $\Hilb_P=\Hilb_F$, the two definitions agree.
\end{definition}

\begin{definition} [Completely positive and trace-preserving maps] \label{definition:cptp}
    Let $\channel\in\Lin(\Lin(\Hilb_0),\Lin(\Hilb_1))$. We say that $\channel$ is completely positive if, for every additional finite-dimensional register $Z$ and every positive semidefinite matrix $X\in\Lin(\Hilb_0\otimes\Hilb_Z)$, the matrix $(\channel\otimes\id_Z)(X)$ is positive semidefinite. Here $\id_Z$ is the identity map on $\Lin(\Hilb_Z)$. We say that $\channel$ is trace-preserving if $\tr(\channel(X))=\tr(X)$ for every $X\in\Lin(\Hilb_0)$. A map satisfying both conditions is called CPTP.
\end{definition}

\begin{definition} [Quantum channel] \label{definition:quantum_channel}
    Let $\Hilb_{0}$, $\Hilb_{1}$ Hilbert spaces and let $\channel \in \Lin\left( \Lin\left( \Hilb_{0} \right) , \Lin\left( \Hilb_{1} \right) \right)$. We say that $\channel$ is a quantum channel from $\Hilb_{0}$ to $\Hilb_{1}$ if it is a completely positive, trace-preserving (linear) map from $\Lin\left( \Hilb_{0} \right)$ to $\Lin\left( \Hilb_{1} \right)$. In case $\Hilb_{0} = \states\left( n \right)$, $\Hilb_{1} = \states\left( m \right)$ for some $n, m \in \Nat$, we say that $\channel$ is a channel from $n$ to $m$ qubits, or an $n$-to-$m$ quantum channel.
\end{definition}

Throughout this work we use the Choi-Jamiolkowski isomorphism (or CJ isomorphism for short), which for Hilbert spaces $\Hilb_{0}$, $\Hilb_{1}$, induces a bijection between
$$
\Lin\left( \Lin\left( \Hilb_{0} \right), \Lin\left( \Hilb_{1} \right) \right)
\enspace ,
$$
which is the set of linear transformations from $\Lin\left( \Hilb_{0} \right)$ to $\Lin\left( \Hilb_{1} \right)$, and
$$
\Lin\left( \Hilb_{0} \otimes \Hilb_{1} \right)
\enspace ,
$$
which is the set of linear transformations from $\Hilb_{0} \otimes \Hilb_{1}$ to itself.
The CJ isomorphism in particular induces an isomorphism between (1) $n$-to-$m$ quantum channels and (2) a special set of matrices in $\bbC^{2^{n + m} \times 2^{n + m}}$. The CJ isomorphism is defined as follows.

\begin{definition} [Choi-Jamiolkowski operator of a linear operator] \label{definition:cj_isomorphism}
    Let $\channel \in \Lin\left( \Lin\left( \Hilb_{0} \right) , \Lin\left( \Hilb_{1} \right) \right)$ a linear operator for $\Hilb_{0} = \states\left( n \right)$, $\Hilb_{1} = \states\left( m \right)$ for some $n, m \in \Nat$. The Choi-Jamiolkowski (or CJ in short) operator of $\channel$ is given by the following matrix $\CJ\left( \channel \right) \in \Lin\left( \Hilb_{0} \otimes \Hilb_{1} \right)$:
    $$
    \CJ\left( \channel \right)
    :=
    \sum_{x, y \in \{ 0, 1 \}^{n} }
    \ket{x}\bra{y} 
    \otimes
    \channel\left( \ket{x}\bra{y}  \right)
    \enspace .
    $$
\end{definition}

We will also use the inverse CJ isomorphism, which is defined as follows.

\begin{definition} [Inverse of Choi-Jamiolkowski transformation] \label{definition:inverse_cj}
    Let $M \in \Lin\left( \Hilb_{0} \otimes \Hilb_{1} \right)$ a linear operator.
    The inverse CJ transformation of the tuple $\left( M, \Hilb_{0}, \Hilb_{1} \right)$ is defined as the linear operator $\CJ^{-1}\left( M, \Hilb_{0}, \Hilb_{1} \right) \in \Lin\left( \Lin\left( \Hilb_{0} \right) , \Lin\left( \Hilb_{1} \right) \right)$ such that for every $X \in \Lin\left( \Hilb_{0} \right)$, 
    $$
    \left[ \CJ^{-1}\left( M, \Hilb_{0}, \Hilb_{1} \right) \right]\left( X \right)
    =
    \tr_{\Hilb_{0}}\left(
    \left( X^{T} \otimes I_{\Lin\left( \Hilb_{1} \right)} \right)
    \cdot
    M
    \right)
    \enspace .
    $$
    Here $X^{T}$ is the transpose in the specified computational basis of $\Hilb_{0}$.
\end{definition}

\begin{remark} [3 input parameters in inverse CJ isomorphism]
    In the above definition \ref{definition:inverse_cj} of the reverse CJ transformation $\CJ^{-1}$, note that unlike the forward transformation $\CJ$, we need to specify not only the operator $M$ but also the spaces $\Hilb_{0}$, $\Hilb_{1}$. The reason is that the space $\Hilb_{0} \otimes \Hilb_{1}$ can be decomposed in many different ways (i.e. for Hilbert spaces of different dimensions), and imply different operators by the inverse CJ transformation. This is different from the case of the forward transformation $\CJ$, where the input and output linear spaces are implicitly defined as part of the operator $\channel$.
\end{remark}

\begin{definition} [Matrix representation of a quantum channel] \label{definition:channel_matrices}
    Let $\channel$ a quantum channel from $\Hilb_0$ to $\Hilb_1$. A matrix representation of $\channel$ is a finite collection of matrices $C_1,\ldots,C_t\in\Lin(\Hilb_0,\Hilb_1)$ such that, for every $X\in\Lin(\Hilb_0)$,
    \begin{equation}
    \channel(X)=\sum_{\alpha=1}^{t}C_{\alpha}XC_{\alpha}^{\dagger} \enspace ,
    \qquad
    \sum_{\alpha=1}^{t}C_{\alpha}^{\dagger}C_{\alpha}=I_{\Lin(\Hilb_0)}
    \enspace .\label{eq:kraus}
    \end{equation}
    Every finite-dimensional quantum channel has such a representation. The matrices $C_{\alpha}$ are also called Kraus operators. They describe the channel mathematically; they are not additional input matrices that an algorithm is required to write down. Omitting the second equality gives a representation of a completely positive map that does not necessarily preserve trace.
\end{definition}

\begin{definition} [Unitary implementation of a quantum channel] \label{definition:channel_dilation}
    Let $\channel$ a quantum channel from $\Hilb_0$ to $\Hilb_1$. A unitary implementation of $\channel$ consists of additional registers $E,E'$ and a unitary $U:\Hilb_0\otimes\Hilb_E\to\Hilb_1\otimes\Hilb_{E'}$ such that, for every input density matrix $\rho$,
    $$
    \channel(\rho)
    =
    \tr_{E'}\left(
    U\left(\rho\otimes\ket{0}\bra{0}_E\right)U^{\dagger}
    \right)
    \enspace .
    $$
    Such an implementation is also called a unitary dilation. For a circuit-described channel, its measurement and discarded registers can be kept coherently to obtain a polynomial-size implementation of this form, which is called the purification of the quantum circuit.
\end{definition}

\begin{definition} [Trace distance and approximation of a state] \label{definition:trace_distance}
    For a matrix $M$, its trace norm is $\norm{M}_1:=\tr\sqrt{M^{\dagger}M}$. For two density matrices $\rho,\sigma$ of the same dimension, their trace distance is
    $$
    \td(\rho,\sigma):=\frac{1}{2}\norm{\rho-\sigma}_1
    \enspace .
    $$
    A state-generation procedure has error at most $\epsilon$ if its output has trace distance at most $\epsilon$ from its specified output state.
\end{definition}

\section{Process Matrix Formalism}\label{section:process_matrix_formalism}
Process matrices were first defined in \cite{oreshkov2012quantum} and since then became a central object in the study of indefinite causal structure in quantum mechanics. In this section we give the required definitions for information processing based on process matrices, computational models, computationally-bounded variations and complexity classes. We follow the definitions from the literature, in some places using previous definitions as they are, and in some other places giving logically equivalent definitions which would be easier to operate in the scope of this paper. We also provide numerous new definitions.

\begin{definition} [Double ket notation] \label{definition:double_ket}
    Let $M \in \Lin\left( \Hilb_{0}, \Hilb_{1} \right)$ a matrix such that $\Hilb_{0} = \states\left( n \right)$, $\Hilb_{1} = \states\left( m \right)$ for some $n, m \in \Nat$. We define the following "double ket" vectors in $\bbC^{ 2^{n + m} }$:
    $$
    \Dket{M}
    :=
    \sum_{ x \in \{ 0, 1 \}^{n} }
    \left(
    \ket{x} \otimes M \cdot \ket{x}
    \right)
    \enspace ,
    $$
    $$
    \Dbra{M}
    :=
    \Dket{M}^{\dagger}
    :=
    \sum_{ x \in \{ 0, 1 \}^{n} }
    \left(
    \bra{x} \otimes \bra{x} \cdot M^{\dagger}
    \right)
    \enspace .
    $$
\end{definition}

\begin{definition} [Boundary pairs and local quantum channels] \label{definition:local_channels}
    Let $k\geq0$, let $\mathbf{m}=(m_1,\ldots,m_k)$ be a tuple of positive integers, and let $n:=\sum_{i=1}^{k}m_i$. For every $i\in[k]$, let $R_i,S_i$ be $m_i$-qubit registers, and set
    $$
    \Hilb_{A_I}=\bigotimes_{i=1}^{k}\Hilb_{R_i},
    \quad
    \Hilb_{A_O}=\bigotimes_{i=1}^{k}\Hilb_{S_i}
    \enspace .
    $$
    The tuple $\mathbf m$ is the ordered partition into boundary pairs. In our circuit notation, the $R_i$ are inputs to the circuit and the $S_i$ are outputs from it. A local quantum channel on pair $i$ is any channel
    $$
    \mathcal A_i:
    \Lin(\Hilb_{S_i}\otimes\Hilb_{E_i})
    \longrightarrow
    \Lin(\Hilb_{R_i}\otimes\Hilb_{E_i'})
    \enspace ,
    $$
    where $E_i,E_i'$ are arbitrary additional registers. These additional registers are external: they are not connected back to the circuit. Their inputs may be entangled with each other and with the external input to the circuit.
\end{definition}

We next define the indefinite operator of a matrix $W$, which is used later to define what is a process matrix. In what follows, we think of $W$ as a CJ representation of some quantum channel from $\Lin \left( \Hilb_{ A_{I} } \otimes \Hilb_{ P } \right)$ to $\Lin \left( \Hilb_{ A_{O} } \otimes \Hilb_{ F } \right)$.

\begin{definition} [Indefinite Operator] \label{definition:indefinite_operator}
    Let $W \in \Lin\left( \Hilb_{ A_{I} } \otimes \Hilb_{ P } \otimes \Hilb_{A_{O}} \otimes \Hilb_{ F } \right)$ for Hilbert spaces $\Hilb_{ A_{I} }$, $\Hilb_{ P }$, $\Hilb_{A_{O}}$, $\Hilb_{ F }$ such that $\Hilb_{ A_{I} } = \Hilb_{ A_{O} } = \states\left( n \right)$. The indefinite operator of $W$ with respect to the Hilbert spaces $\Hilb_{ A_{I} }$, $\Hilb_{ A_{O} }$ is the operator $G_{W, A_{I}, A_{O}} \in \Lin\left( \Hilb_{ P } \otimes \Hilb_{ F } \right)$ such that
    $$
    G_{ W, A_{I}, A_{O} }
    :=
    \tr_{ A_{I}, A_{O} }
    \left(
    W
    \cdot
    \ket{ A_{ I }, P, F }
    \bra{ A_{ I }, P, F }
    \right)
    \enspace ,
    $$
    where 
    $$
    \ket{ A_{ I }, P, F }
    :=
    \sum_{ x \in \{ 0, 1 \}^{n} }
    \ket{x} \otimes I_{ \Lin\left( \Hilb_{P} \right) } \otimes \ket{x} \otimes I_{ \Lin\left( \Hilb_{F} \right) }
    \enspace ,
    $$
    $$
    \bra{ A_{ I }, P, F }
    :=
    \sum_{ y \in \{ 0, 1 \}^{n} }
    \bra{y} \otimes I_{ \Lin\left( \Hilb_{P} \right) } \otimes \bra{y} \otimes I_{ \Lin\left( \Hilb_{F} \right) }
    \enspace .
    $$
\end{definition}

\paragraph{Connecting local operations to the feedback registers.}
We next allow a quantum channel to be inserted into each feedback pair. Each local channel receives a feedback output and returns a feedback input; its auxiliary registers remain external. The following definition gives the CJ operator of the resulting transformation. It uses the contraction from \cite[Section III]{araujo2017quantum}, with our input and output labels.

\begin{definition}[Indefinite operator with local channels]
\label{definition:extended_channel_indefinite_operator}
    Let
    $$
    W\in\Lin\left(
    \Hilb_{A_I}\otimes\Hilb_P\otimes\Hilb_{A_O}\otimes\Hilb_F
    \right),
    $$
    and fix a partition into feedback pairs as in Definition \ref{definition:local_channels}, with
    $$
    \Hilb_{A_I}=\bigotimes_{i=1}^{k}\Hilb_{R_i}
    \quad \text{and} \quad
    \Hilb_{A_O}=\bigotimes_{i=1}^{k}\Hilb_{S_i}
    \enspace .
    $$
    For each $i\in[k]$, let $E_i,E_i'$ be auxiliary registers and let
    $$
    \mathcal A_i:
    \Lin\left(\Hilb_{S_i}\otimes\Hilb_{E_i}\right)
    \longrightarrow
    \Lin\left(\Hilb_{R_i}\otimes\Hilb_{E_i'}\right)
    $$
    be a quantum channel. Write
    $\boldsymbol{\mathcal A}:=\left(\mathcal A_1,\ldots,\mathcal A_k\right)$,
    $E:=\left(E_1,\ldots,E_k\right)$, and
    $E':=\left(E_1',\ldots,E_k'\right)$.

    The indefinite operator of $W$ with these local channels is
    $$
    G_{W,A_I,A_O}\left[\boldsymbol{\mathcal A}\right]
    :=
    \tr_{A_I,A_O}\left[
    \left(W^{T_{A_I,A_O}}\otimes I_{\Hilb_E\otimes\Hilb_{E'}}\right)
    \left(I_{\Hilb_P\otimes\Hilb_F}\otimes
    \bigotimes_{i=1}^{k}\CJ\left(\mathcal A_i\right)\right)
    \right]
    \enspace .
    $$
    Here $T_{A_I,A_O}$ is the partial transpose on the feedback registers in their computational bases. Each factor acts on its named registers, with tensor factors reordered accordingly. After taking the partial trace, the remaining registers are ordered as $P,E,F,E'$.

    The induced linear map from $P,E$ to $F,E'$ is
    \begin{equation}
    \pmo\left(W,\Hilb_{A_I},\Hilb_{A_O}\right)
    \left[\boldsymbol{\mathcal A}\right]
    :=
    \CJ^{-1}\left(
    G_{W,A_I,A_O}\left[\boldsymbol{\mathcal A}\right],
    \Hilb_P\otimes\Hilb_E,
    \Hilb_F\otimes\Hilb_{E'}
    \right)
    \enspace .
    \label{eq:connection}
    \end{equation}
\end{definition}

When all local channels are identities and there are no auxiliary registers, this reduces to the indefinite operator in Definition \ref{definition:indefinite_operator}. The induced map need not be a quantum channel for an arbitrary $W$; requiring this for every choice of local channels is part of the process matrix definition below.

We define process matrices as follows. 

\begin{definition} [Process Matrix] \label{definition:process_matrix}
    Let $W \in \Lin\left( \Hilb_{ A_{I} } \otimes \Hilb_{ P } \otimes \Hilb_{A_{O}} \otimes \Hilb_{ F } \right)$ for Hilbert spaces $\Hilb_{ A_{I} }$, $\Hilb_{ P }$, $\Hilb_{A_{O}}$, $\Hilb_{ F }$ such that $\Hilb_{ A_{I} } = \Hilb_{ A_{O} } = \states\left( n \right)$. Define the following linear operator $\pmo\left( W, \Hilb_{ A_{I} }, \Hilb_{ A_{O} } \right) \in \Lin\left( \Lin\left( \Hilb_{ P } \right), \Lin\left( \Hilb_{ F } \right) \right)$ as
    $$
    \pmo\left( W, \Hilb_{ A_{I} }, \Hilb_{ A_{O} } \right) 
    :=
    \CJ^{-1}\left( G_{ W, \Hilb_{ A_{I} }, \Hilb_{ A_{O} } }, \; \Hilb_{ P }, \Hilb_{ F } \right)
    \enspace ,
    $$
    where $G_{ W, \Hilb_{ A_{I} }, \Hilb_{ A_{O} } }$ is the indefinite operator of $W$ with respect to the spaces $\Hilb_{ A_{I} }$, $\Hilb_{ A_{O} }$, as in Definition \ref{definition:indefinite_operator}, and $\CJ^{-1}$ is the inverse CJ transformation, as in Definition \ref{definition:inverse_cj}.
    
    Fix a partition $\mathbf{m}$ of the two boundary registers as in Definition \ref{definition:local_channels}. We say that $W$ is a process matrix with respect to $\Hilb_{ A_{I} }$, $\Hilb_{ A_{O} }$, and $\mathbf{m}$ if $W$ is positive semidefinite and, for every collection of auxiliary registers and every choice of local quantum channels $\mathcal{A}_{1},\ldots,\mathcal{A}_{k}$, the induced operator
    $$
    \pmo\left( W, \Hilb_{ A_{I} }, \Hilb_{ A_{O} } \right)[\boldsymbol{\mathcal{A}}]
    $$
    from Definition \ref{definition:extended_channel_indefinite_operator} is a quantum channel.
    From hereon unless noted otherwise the partition $\mathbf{m}$ is implicitly set to be $(m_1, \ldots, m_{k})$, where $k=n$ and such that for all $i \in [k]$, $m_{i} := 1$, which is why we suppress the notation. 
\end{definition}

\begin{remark} [Intuition behind the operator $\pmo$]
    The name $\pmo$ is meant for a Process Matrix Operator. The idea behind the operator $\pmo$ is that it tries to take a general linear map, and map it to not just any linear transformation, but a quantum channel. What defines a process matrix $W$ is that it is positive semidefinite and that this mapping gives a quantum channel for every choice of local quantum channels on the specified boundary pairs, including additional registers. The identity connection is only one of these choices.
\end{remark}

\subsection{A Computational Model based on Process Matrices}
In the above subsections, all definitions were information theoretical. The work \cite{araujo2017quantum} was the first to consider the complexity theoretical aspect of process matrices. In this subsection we add our own notions of computation and computational complexity of process matrices, based on the previous work.

\begin{definition} [Quantum channel of a quantum circuit]
    For $n, m \in \Nat$, let $C$ a quantum circuit from $n$ to $m$ qubits (as in defined in Section \ref{section:preliminaries}).

    The quantum channel of $C$, denoted $\channel_{C} \in \Lin\left( \Lin\left( \states\left( n \right) \right), \Lin\left( \states\left( m \right) \right) \right)$ is the linear map that agrees with the execution of $C$ on every $n$-qubit mixed quantum state.
\end{definition}

\begin{definition} [Circuit description for a process] \label{definition:process_descriptor}
    Let $n$ a non-negative integer, $C$ an $(n+r)$-to-$(n+\ell)$ quantum circuit with the boundary registers ordered as in Definition \ref{definition:local_channels}. The pair $D:=(n,C)$ is a circuit description for a process. We implicitly use $n$ separate single-qubit boundary pairs $(R_i,S_i)$, $i\in[n]$, where $R_i$ is an input qubit of $C$ and $S_i$ is the corresponding output qubit.
    
    For any $n$ arbitrary local channels $\mathcal A_1$, $\ldots$, $\mathcal A_n$, the channel $\mathcal A_i$ acts from $S_iE_i$ to $R_iE_i'$, where $E_i,E_i'$ may have arbitrary sizes and remain external. Their inputs may be entangled with each other and with the external input to $C$. For $n=0$, there are no boundary pairs. The size $|D|$ includes $n$ and $|C|$.
\end{definition}

\paragraph{Quantum Circuits as Process Matrix Generators.}
We next define a new notion, called Process Matrix Generators (PMG). Intuitively, PMGs are quantum circuits such that if we take their CJ representations, we get a process matrix. More precisely, for an $(n + r)$-to-$(n + \ell)$ quantum circuit $C$ we take the $n$ first qubits to be boundary registers and ask whether the CJ representation is a process matrix w.r.t. these registers. We care about polynomial-time PMGs, as they play a central role in the definition of polynomial-time process matrix computation. In particular, when we later define the complexity class $\BQPlctc$ from previous work (and a new class $\BQPuctc$), we will do so using the new definition, and ask that all process matrix transformations allowed in $\BQPlctc$ are only by quantum polynomial-time PMGs.

\begin{definition} [Process Matrix Generator] \label{definition:process_matrix_generator}
    Let $n \in \Nat\cup\{0\}$, let $r, \ell \in \left( \Nat \cup \{ 0 \} \right)$ and let $C$ an $(n + r)$-to-$(n + \ell)$ quantum circuit. Consider the quantum channel $\channel_{C} \in \Lin\left( \Lin\left( \Hilb_{ A_{I} } \otimes \Hilb_{P} \right), \Lin\left( \Hilb_{A_{O}} \otimes \Hilb_{F} \right) \right)$ of the circuit $C$, such that $\Hilb_{ A_{I} } = \states\left( n \right)$, $\Hilb_{ P } = \states\left( r \right)$, $\Hilb_{ A_{O} } = \states\left( n \right)$, $\Hilb_{ F } = \states\left( \ell \right)$.
    
    Define the following operator $\pmo_{\left( n, C \right)} \in \Lin\left( \Lin\left( \Hilb_{P} \right), \Lin\left( \Hilb_{F} \right) \right)$:
    $$
    \pmo_{\left( n, C \right)}
    :=
    \pmo\left( \CJ\left( \channel_{C} \right), \Hilb_{ A_{I} }, \Hilb_{ A_{O} } \right)
    \enspace ,
    $$
    where $\pmo\left( \CJ\left( \channel_{C} \right), \Hilb_{ A_{I} }, \Hilb_{ A_{O} } \right)$ is as per Definition \ref{definition:process_matrix}.
    We use the implicit single-qubit boundary convention from Definition \ref{definition:process_descriptor}. If $\CJ\left(\channel_C\right)$ is a process matrix for these $n$ separate boundary pairs, as in Definition \ref{definition:process_matrix}, we say that the tuple $\left( n, C \right)$ is a Process Matrix Generator (or PMG), and denote by $\pmo_{\left( n, C \right)}$ the Process Matrix Channel of the tuple.
    For $D=(n,C)$, we also write $\pmo_D:=\pmo_{(n,C)}$ and $\pmo_D[\boldsymbol{\mathcal A}]$ for the connection with local channels.
\end{definition}

\paragraph{Simplified formula for evaluating process matrix generators.}
The following Lemma gives a formula for the identity-connected operator of the tuple $\left( n, C \right)$. It does not replace the validity condition for all local quantum channels.
Specifically, assume that $C$ is an isometry on quantum states, that given an $(n + r)$-qubit pure state $\ket{\psi} \in \states\left( n + r \right)$, concatenates $k$ ancillary qubits and executes a unitary $U$ on $n + r + k$ qubits. The below Lemma is a generalization of the simplifying Equation (21) from \cite{araujo2017quantum}.

\begin{lemma} [Simplified PM Formula for Isometries] \label{lemma:simplified_isometry_pm}
    Let $n, r \in \Nat$, $k \in \Nat \cup \{ 0 \}$, and let $U \in \bbC^{ \left( 2^{ n + r + k } \right) \times \left( 2^{ n + r + k } \right) }$ a unitary matrix.
    Let $\channel_{\left( U, k \right)}$ the quantum channel of the $k$-isometry of $U$, that is, for an input $(n + r)$-qubit mixed state $\rho$ the channel adds an ancilla of $k$ qubits and executes $U$, that is,
    $$
    \channel_{\left( U, k \right)}\left( \rho \right)
    :=
    U \cdot \left( \rho \otimes \ket{ 0^{k} }\bra{ 0^{k} } \right) \cdot U^{\dagger}
    \enspace .
    $$
    Denote the Hilbert spaces of the channel,
    $$
    \channel_{\left( U, k \right)} \in \Lin\left( \Lin\left( \Hilb_{ A_{I} } \otimes \Hilb_{P} \right), \Lin\left( \Hilb_{A_{O}} \otimes \Hilb_{F} \right) \right) \enspace ,
    $$
    such that $\Hilb_{ A_{I} } = \states\left( n \right)$, $\Hilb_{ P } = \states\left( r \right)$, $\Hilb_{ A_{O} } = \states\left( n \right)$, $\Hilb_{ F } = \states\left( r + k \right)$. 

    Let $W := \CJ\left( \channel_{\left( U, k \right)} \right)$ and let $G_{W, A_{I}, A_{O}}$ the indefinite operator of $W$ with respect to $\Hilb_{ A_{I} }$, $\Hilb_{ A_{O} }$ (as in Definition \ref{definition:indefinite_operator}). Then we have the equality
    $$
    G_{W, A_{I}, A_{O}}
    =
    \CJ\left( \channel_{ \left( U_{G}, k \right) } \right)
    \enspace ,
    $$
    for the matrix $U_{G} \in \bbC^{ 2^{r + k} \times 2^{r + k} }$ such that
    $$
    U_{G}
    =
    \tr_{ \Hilb_{ A_{I} } }\left( U \right) \enspace .
    $$
    That is, $U_{G}$ is the matrix derived by tracing out the Hilbert space of the leftmost $n$ qubits from the unitary matrix $U$ itself, and the linear operator $\channel_{ \left( U_{G}, k \right) }$ is such that for an input matrix $\sigma \in \bbC^{ 2^{r} \times 2^{r} }$, outputs:
    $$
    \channel_{ \left( U_{G}, k \right) }\left( \sigma \right)
    =
    U_{G} \cdot \left( \sigma \otimes \ket{ 0^{k} }\bra{ 0^{k} } \right) \cdot U^{\dagger}_{G}
    \enspace .
    $$

    In particular, if $W$ is a process matrix with respect to $\Hilb_{A_I}$ and $\Hilb_{A_O}$ under the implicit single-qubit boundary convention, then
    \begin{itemize}
        \item
        $\channel_{\left(U_G,k\right)}$ is a valid quantum channel.

        \item
        For any quantum circuit $C$ that implements the original channel $\channel_{\left(U,k\right)}$ with the same ordered boundary registers, the tuple $\left(n,C\right)$ is a PMG with generated channel $\pmo_{(n,C)}=\channel_{\left(U_G,k\right)}$ (as in Definition \ref{definition:process_matrix_generator}).
    \end{itemize}
    Note that the converse does not follow from the identity-connected formula: $\channel_{\left(U_G,k\right)}$ can be a quantum channel even when $W$ is not a process matrix.
\end{lemma}

\begin{proof}
     Note that there is an isometry $V \in \bbC^{ 2^{n + r + k} \times 2^{n + r} }$ that for any $x \in \{ 0, 1 \}^{n + r}$, $V\cdot \ket{x} = U \cdot \left( \ket{x} \otimes \ket{ 0^{k} } \right)$. It follows that the action of the channel $\channel_{\left( U, k \right)}$ for an input $(n + r)$-qubit mixed state $\rho$ is $V \cdot \rho \cdot V^{\dagger}$. This implies that $W := \CJ\left( \channel_{\left( U, k \right)} \right) = \Dket{V}\Dbra{V}$.
    
    By Definition \ref{definition:indefinite_operator}, the indefinite operator $G_{W, A_{ I }, A_{ O }}$ of $W = \Dket{V}\Dbra{V}$ with respect to the Hilbert spaces $\Hilb_{ A_{I} }$, $\Hilb_{ A_{O} }$ is
    $$
    G_{ W, A_{I}, A_{O} }
    :=
    \tr_{ A_{I}, A_{O} }
    \left(
    \Dket{V}\Dbra{V}
    \cdot
    \ket{ A_{I}, P, F } \bra{ A_{I}, P, F }
    \right)
    \enspace .
    $$

    Let us break down the expression inside the partial trace. By Definition \ref{definition:double_ket}, 
    $$
    \Dket{V}\Dbra{V}
    :=
    \sum_{ x \in \{ 0, 1 \}^{n + r} } \left( \ket{x} \otimes V \cdot \ket{x} \right)
    \sum_{ y \in \{ 0, 1 \}^{n + r} } \left( \bra{y} \otimes \bra{y} \cdot V^{\dagger} \right)
    \enspace .
    $$
    From Definition \ref{definition:indefinite_operator} we also have
    $$
    \ket{ A_{ I }, P, F }
    :=
    \sum_{ u \in \{ 0, 1 \}^{n} }
    \ket{u} \otimes I_{ \Lin\left( \Hilb_{P} \right) } \otimes \ket{u} \otimes I_{ \Lin\left( \Hilb_{F} \right) }
    \enspace ,
    $$
    $$
    \bra{ A_{ I }, P, F }
    :=
    \sum_{ v \in \{ 0, 1 \}^{n} }
    \bra{v} \otimes I_{ \Lin\left( \Hilb_{P} \right) } \otimes \bra{v} \otimes I_{ \Lin\left( \Hilb_{F} \right) }
    \enspace .
    $$

    By the linearity of trace we get
    $$
    G_{ W, A_{I}, A_{O} }
    :=
    \tr_{ A_{I}, A_{O} }
    \left(
    \Dket{V}\Dbra{V}
    \cdot
    \ket{ A_{I}, P, F } \bra{ A_{I}, P, F }
    \right)
    $$
    $$
    =
    \sum_{ x, y \in \{ 0, 1 \}^{n + r}, u, v \in \{ 0, 1 \}^{n} }
    \tr_{ A_{I}, A_{O} }
    \bigg(
    \left( \ket{x} \otimes V \cdot \ket{x} \right)
    \left( \bra{y} \otimes \bra{y} \cdot V^{\dagger} \right)
    $$
    $$
    \hspace{4cm}
    \cdot \left(
    \ket{u} \otimes I_{ \Lin\left( \Hilb_{P} \right) } \otimes \ket{u} \otimes I_{ \Lin\left( \Hilb_{F} \right) }
    \right)
    \left(
    \bra{v} \otimes I_{ \Lin\left( \Hilb_{P} \right) } \otimes \bra{v} \otimes I_{ \Lin\left( \Hilb_{F} \right) }
    \right)
    \bigg)
    \enspace .
    $$

    Next, by the cyclic property of trace, the above equals 
    $$
    =
    \sum_{ x, y \in \{ 0, 1 \}^{n + r}, u, v \in \{ 0, 1 \}^{n} }
    \tr_{ A_{I}, A_{O} }
    \bigg(
    \left(
    \bra{v} \otimes I_{ \Lin\left( \Hilb_{P} \right) } \otimes \bra{v} \otimes I_{ \Lin\left( \Hilb_{F} \right) }
    \right)
    \left( \ket{x} \otimes V \cdot \ket{x} \right)
    $$
    $$
    \hspace{4cm}
    \cdot
    \left( \bra{y} \otimes \bra{y} \cdot V^{\dagger} \right)
    \left(
    \ket{u} \otimes I_{ \Lin\left( \Hilb_{P} \right) } \otimes \ket{u} \otimes I_{ \Lin\left( \Hilb_{F} \right) }
    \right)
    \bigg)
    $$
    $$
    =
    \sum_{ a, b \in \{ 0, 1 \}^{n}, i, j \in \{ 0, 1 \}^{r}, u, v \in \{ 0, 1 \}^{n} }
    \tr_{ A_{I}, A_{O} }
    \bigg(
    \left(
    \bra{v} \otimes I_{ \Lin\left( \Hilb_{P} \right) } \otimes \bra{v} \otimes I_{ \Lin\left( \Hilb_{F} \right) }
    \right)
    \left( \ket{a, i} \otimes V \cdot \ket{a, i} \right)
    $$
    $$
    \hspace{4cm}
    \cdot
    \left( \bra{b, j} \otimes \bra{b, j} \cdot V^{\dagger} \right)
    \left(
    \ket{u} \otimes I_{ \Lin\left( \Hilb_{P} \right) } \otimes \ket{u} \otimes I_{ \Lin\left( \Hilb_{F} \right) }
    \right)
    \bigg)
    \enspace .
    $$

    By considering the result of inner products for the expression inside the trace, we get that the summand for $i, j, u, v$ is
    $$
    \left(
    \ket{i}
    \otimes
    \left(
    \bra{v} \otimes I_{ \Lin\left( \Hilb_{F} \right) }
    \right)
    \cdot 
    V \cdot \ket{v,i}
    \right)
    \cdot
    \left(
    \bra{j}
    \otimes
    \bra{u, j} \cdot V^{\dagger}
    \cdot
    \left(
    \ket{u} \otimes I_{ \Lin\left( \Hilb_{F} \right) }
    \right)
    \right)
    \enspace .
    $$
    Observe that at this point of the calculation, for the expression inside the trace (for each set of indices $i, j, u, v$), the systems $\Hilb_{ A_{I} }$, $\Hilb_{ A_{O} }$ are not longer present. In particular, the trace can be removed from the previous sum, and we get that $G_{ W, A_{I}, A_{O} }$ equals:
    $$
    =
    \sum_{ i, j \in \{ 0, 1 \}^{r}, u, v \in \{ 0, 1 \}^{n} }
    \left(
    \ket{i}
    \otimes
    \left(
    \bra{v} \otimes I_{ \Lin\left( \Hilb_{F} \right) }
    \right)
    \cdot 
    V \cdot \ket{v,i}
    \right)
    \cdot
    \left(
    \bra{j}
    \otimes
    \bra{u, j} \cdot V^{\dagger}
    \cdot
    \left(
    \ket{u} \otimes I_{ \Lin\left( \Hilb_{F} \right) }
    \right)
    \right)
    \enspace .
    $$

    We next consider, for every \emph{fixed} pair $i, j \in \{ 0, 1 \}^{r}$, the sum over all \emph{different} $u, v \in \{ 0, 1 \}^{n}$:
    $$
    \sum_{ v, u \in \{ 0, 1 \}^{n} }
    \left(
    \ket{i}
    \otimes
    \left(
    \left(
    \bra{v} \otimes I_{ \Lin\left( \Hilb_{F} \right) }
    \right)
    \cdot 
    V \cdot \ket{v,i}
    \right)
    \right)
    \cdot
    \left(
    \bra{j}
    \otimes
    \left(
    \bra{u, j} \cdot V^{\dagger}
    \cdot
    \left(
    \ket{u} \otimes I_{ \Lin\left( \Hilb_{F} \right) }
    \right)
    \right)
    \right)
    \enspace .
    $$
    Observe that in the above we have a sum, over $u, v \in \{ 0, 1 \}^{n}$, where every summand is a product between two elements. However, the left element in each product only depends on $v$, and every element on the right only depends on $u$. It follows that the above can be written in the form:
    $$
    =
    \left(
    \sum_{ v \in \{ 0, 1 \}^{n} }
    \ket{i}
    \otimes
    \left(
    \left(
    \bra{v} \otimes I_{ \Lin\left( \Hilb_{F} \right) }
    \right)
    \cdot 
    V \cdot \ket{v,i}
    \right)
    \right)
    \cdot
    \left(
    \sum_{ u \in \{ 0, 1 \}^{n} }
    \bra{j}
    \otimes
    \left(
    \bra{u, j} \cdot V^{\dagger}
    \cdot
    \left(
    \ket{u} \otimes I_{ \Lin\left( \Hilb_{F} \right) }
    \right)
    \right)
    \right)
    \enspace .
    $$
    Let us analyze the left multiplicand, and the analysis of the right multiplicand will be symmetric:
    $$
    \sum_{ v \in \{ 0, 1 \}^{n} }
    \ket{i}
    \otimes
    \left(
    \left(
    \bra{v} \otimes I_{ \Lin\left( \Hilb_{F} \right) }
    \right)
    \cdot 
    V \cdot \ket{v,i}
    \right)
    $$
    $$
    =
    \ket{i}
    \otimes
    \left(
    \sum_{ v \in \{ 0, 1 \}^{n} }
    \left(
    \left(
    \bra{v} \otimes I_{ \Lin\left( \Hilb_{F} \right) }
    \right)
    \cdot 
    V \cdot \ket{v,i}
    \right)
    \right)
    \enspace .
    $$

    Recall that $\forall (v, i) \in \{ 0, 1 \}^{n + r}$ we have $V\cdot \ket{v, i} = U \cdot \left( \ket{v, i} \otimes \ket{0^k} \right)$. Now, observe that for every $i \in \{ 0, 1 \}^{r}$ we have 
    $$
    \sum_{ v \in \{ 0, 1 \}^{n} }
    \left(
    \left(
    \bra{v} \otimes I_{ \Lin\left( \Hilb_{F} \right) }
    \right)
    \cdot 
    V \cdot \ket{v,i}
    \right)
    $$
    $$
    =
    \sum_{ v \in \{ 0, 1 \}^{n} }
    \left(
    \left(
    \bra{v} \otimes I_{ \Lin\left( \Hilb_{F} \right) }
    \right)
    \cdot 
    U \cdot \left( \ket{v, i} \otimes \ket{0^k} \right)
    \right)
    $$
    $$
    =
    \left( \tr_{ A_{I} }\left( U \right) \right)
    \cdot
    \left( \ket{i} \otimes \ket{0^k} \right)
    \enspace .
    $$
    To conclude, by denoting $U_{G} := \tr_{ A_{I} }\left( U \right) \in \bbC^{ 2^{ r + k } \times 2^{ r + k } }$, for every $i \in \{ 0, 1 \}^{r}$ we have the equality :
    $$
    \sum_{ v \in \{ 0, 1 \}^{n} }
    \ket{i}
    \otimes
    \left(
    \left(
    \bra{v} \otimes I_{ \Lin\left( \Hilb_{F} \right) }
    \right)
    \cdot 
    V \cdot \ket{v,i}
    \right)
    =
    \ket{i}
    \otimes
    U_{G}
    \cdot
    \left( \ket{i} \otimes \ket{0^k} \right)
    \enspace .
    $$
    In a similar manner, by a symmetric calculation we have the equality for every $j \in \{ 0, 1 \}^{r}$:
    $$
    \sum_{ u \in \{ 0, 1 \}^{n} }
    \bra{j}
    \otimes
    \left(
    \bra{u, j} \cdot V^{\dagger}
    \cdot
    \left(
    \ket{u} \otimes I_{ \Lin\left( \Hilb_{F} \right) }
    \right)
    \right)
    =
    \bra{j}
    \otimes
    \left( \bra{j} \otimes \bra{0^k} \right)
    \cdot
    U^{\dagger}_{G}
    \enspace .
    $$

    Let us summarize the results and conclude the proof. We saw that 
    $$
    G_{ W, A_{I}, A_{O} }
    =
    \sum_{ i, j, \in \{ 0, 1 \}^{r} }
    \left(
    \sum_{ v \in \{ 0, 1 \}^{n} }
    \ket{i}
    \otimes
    \left(
    \left(
    \bra{v} \otimes I_{ \Lin\left( \Hilb_{F} \right) }
    \right)
    \cdot 
    V \cdot \ket{v,i}
    \right)
    \right)
    $$
    $$
    \hspace{6cm}
    \cdot
    \left(
    \sum_{ u \in \{ 0, 1 \}^{n} }
    \bra{j}
    \otimes
    \left(
    \bra{u, j} \cdot V^{\dagger}
    \cdot
    \left(
    \ket{u} \otimes I_{ \Lin\left( \Hilb_{F} \right) }
    \right)
    \right)
    \right)
    \enspace ,
    $$
    and by our analysis of the multiplicands in the sum, the above equals
    $$
    =
    \sum_{ i, j, \in \{ 0, 1 \}^{r} }
    \left(
    \ket{i}
    \otimes
    U_{G}
    \cdot
    \left( \ket{i} \otimes \ket{0^k} \right)
    \right)
    \left(
    \bra{j}
    \otimes
    \left( \bra{j} \otimes \bra{0^k} \right)
    \cdot
    U^{\dagger}_{G}
    \right)
    $$
    $$
    =
    \sum_{ i, j, \in \{ 0, 1 \}^{r} }
    \ket{i}\bra{j}
    \otimes
    \left(
    U_{G}
    \cdot
    \left( \ket{i} \otimes \ket{0^k} \right)
    \right)
    \left(
    \left( \bra{j} \otimes \bra{0^k} \right)
    \cdot
    U^{\dagger}_{G}
    \right)
    $$
    $$
    =
    \sum_{ i, j, \in \{ 0, 1 \}^{r} }
    \ket{i}\bra{j}
    \otimes
    \left(
    U_{G}
    \cdot
    \left( \ket{i}\bra{j} \otimes \ket{0^k}\bra{0^k} \right)
    \cdot
    U^{\dagger}_{G}
    \right)
    \enspace .
    $$
    It is left to observe that by the Definition \ref{definition:cj_isomorphism}, the last expression is exactly the CJ transformation of the linear operator $\channel_{ \left( U_{G}, k \right) }$ defined in the statement of Lemma \ref{lemma:simplified_isometry_pm}.
\end{proof}

\subsection{Existing and New Complexity Classes}\label{sec:models}
We now define two models of efficient process matrix computation. The first allows general generating circuits, provided they describe valid processes. The second requires the generating circuits to be unitary and the processes to be pure. We also distinguish this restriction from purifiability, which requires a pure extension to exist without imposing an efficiency requirement on its construction.

\paragraph{General process computation.}
Both our model and that of \cite[Section IV]{araujo2017quantum} assume access to the channel specified by a classical generating-circuit description. In the CTC interpretation, the model grants the ability to construct the corresponding feedback connections. The previous definition requires the supplied circuit to describe a valid process. Here we also specify what happens when it does not: a well-formed generation query returns an arbitrary quantum channel of the prescribed dimensions, and the algorithm must satisfy its correctness requirement for every such channel. This defines the generation operation on every well-formed description, without requiring the algorithm to determine process validity.

\begin{definition} [Linear Closed Timelike Curve Generator] \label{definition:linear_ctc_generator}
    Let $A$ a function with input set of pairs $\left( m, C \right)$ where $m \in \Nat\cup\{0\}$ and $C$ a quantum circuit from $m + r$ to $m + \ell$ qubits for some $r, \ell \in \Nat\cup\{0\}$, and output set of $r$-to-$\ell$ qubit quantum channels. $A$ need not be efficiently computable. We say that $A$ is a linear CTC generator if for every well-formed input $\left( m, C \right)$, if the pair is a PMG (as in Definition \ref{definition:process_matrix_generator}) then $A$ outputs the channel $\pmo_{\left( m, C \right)}$, and if the input is not a PMG, the output channel of $A$ can be arbitrary.
\end{definition}

We can now define the complexity class $\BQPlctc$.

\begin{definition} [The Complexity Class $\BQPlctc$] \label{definition:BQPlctc}
    The complexity class $\BQPlctc$ is the set of promise problems $\prod := \left( \YES , \NO \right)$ that are solvable by a quantum polynomial-time algorithm with access to linear post-selection CTCs. Formally, such that there exists a quantum polynomial-time Turing Machine $M$ such that for every linear CTC generator $A$ (as per Definition \ref{definition:linear_ctc_generator}), for every $x \in \YES$, $M^{A}(x) = 1$ with probability at least $p(x)$, and for every $x \in \NO$, $M^{A}(x) = 0$ with probability at least $p(x)$, such that
    \begin{itemize}
        \item
        $$
        \forall x \in \left( \YES \cup \NO \right) : p(x) \geq 1 - 2^{-|x|} \enspace ,
        $$

        \item 
        $M^{A}$ is the machine $M$ with oracle access to $A$ and the channels that it outputs.

        \item
        Every generation query $\left(m,C\right)$ has a classical description. The query may depend on the input $x$ itself and on the preceding computation, and may be adaptive.
        
        \item 
        Every query $\left( m, C \right)$ to $A$ and every use of the channel it produces costs $\Theta\left( |C| \right)$ time.
    \end{itemize}
\end{definition}

\paragraph{Pure processes.}
Ara\'{u}jo, Feix, Navascu\'{e}s and Brukner \cite[Definition 1]{araujo2017purification} introduced pure processes as those that induce a global unitary transformation whenever the local operations are unitary, including operations on auxiliary registers. We give the relevant definitions.

\begin{definition}[Unitary Quantum Channel] \label{definition:unitary_quantum_channel}
    Let $n\in\Nat\cup\{0\}$ and let $\channel$ be an $n$-to-$n$ qubit quantum channel, as in Definition \ref{definition:quantum_channel}. We say that $\channel$ is a unitary quantum channel if there exists a unitary matrix $U_{\channel}\in\bbC^{2^n\times 2^n}$ such that, for every $n$-qubit pure state $\ket{\psi}$,
    $$
    \channel\left( \ket{\psi}\bra{\psi} \right)
    =
    U_{\channel}\cdot\ket{\psi}\bra{\psi}\cdot U_{\channel}^{\dagger}
    \enspace .
    $$
\end{definition}

\begin{definition}[Pure Process Matrix] \label{definition:pure_process_matrix}
    Let $W$ be a process matrix as in Definition \ref{definition:process_matrix}, with feedback input $R$, feedback output $S$, external input $P$, and external output $F$, with a fixed partition into feedback pairs $\left( R_i, S_i \right)$, $i\in[k]$.

    We say that $W$ is \emph{pure} if, for every choice of auxiliary registers $E_i, E_i'$ and local unitary channels $\mathcal A_i$ from $S_i, E_i$ to $R_i, E_i'$, the induced channel
    $$
    \pmo\left( W,\Hilb_R,\Hilb_S \right)
    \left[ \mathcal A_1,\ldots,\mathcal A_k \right]
    $$
    from $P, E_1,\ldots,E_k$ to $F, E_1',\ldots,E_k'$ is unitary.
\end{definition}

By Theorem 2 of \cite{araujo2017purification}, a process matrix $W$ is pure if and only if
$$
W=\Dket{U}\Dbra{U}
$$
for a unitary $U:\Hilb_R\otimes\Hilb_P\longrightarrow\Hilb_S\otimes\Hilb_F$. We use this notion to define pure process matrix generators.

\begin{definition}[Pure Process Matrix Generator] \label{definition:pure_process_matrix_generator}
    Let $D=\left( n,C \right)$ be a PMG as in Definition \ref{definition:process_matrix_generator}. We say that $D$ is a \emph{pure process matrix generator}, or pure PMG, if $C$ is a unitary quantum circuit. Each of the following conditions is equivalent to $D$ being pure:
    \begin{itemize}
        \item
        $\CJ\left( \channel_C \right)$ is a pure process matrix according to Definition \ref{definition:pure_process_matrix}, with the same $n$ separate single-qubit feedback pairs.

        \item
        $C$ is a unitary circuit with matrix $U_C$ satisfying the following condition. Write $R:=\left( R_1,\ldots,R_n \right)$ for the feedback registers. For every choice of auxiliary registers $E:=\left( E_1,\ldots,E_n \right)$ and local unitaries $V_i$ acting on $R_i, E_i$,
        \begin{equation}
            \tr_R\left[
            \left( \bigotimes_{i=1}^{n}V_i \right)
            \left( U_C\otimes I_{\Hilb_E} \right)
            \right]
            \quad\text{is unitary}.
            \label{eq:pure-test}
        \end{equation}
        Here we identify corresponding feedback input and output registers, and each operator acts as the identity on registers not mentioned in its definition. The auxiliary registers remain external and are not traced out.
    \end{itemize}
    In particular, the generated channel $\pmo_D$ is unitary.
\end{definition}

\begin{remark}[A PMG with a unitary generated channel is not necessarily pure]
    A unitary generated channel alone does not make a PMG pure. For example, our Choi-state construction in Lemma \ref{lemma:from_choi_to_isometry} generates a unitary channel when $V$ is unitary, although the ordinary channel $\channel_C$ implemented by its generating circuit is not unitary. To see which local operations reveal this distinction, introduce an auxiliary input qubit $E_i$ at each feedback pair $\left( R_i, S_i \right)$ and choose the local operation to be a SWAP. This sends $E_i$ into $R_i$ and returns $S_i$ as an auxiliary output $E_i'$. The resulting global channel is exactly $\channel_C$, up to register names. Thus, although every inserted operation is unitary, the global channel is not. Purity requires unitarity for every such choice, not only when the local operations are identities.
\end{remark}

\paragraph{Pure process computation.}
We define the class $\BQPuctc$, which restricts $\BQPlctc$ to pure process computation.

\begin{definition} [Unitary Closed Timelike Curve Generator] \label{definition:unitary_ctc_generator}
    Let $A$ a function with input set of pairs $\left( m, C \right)$ where $m \in \Nat\cup\{0\}$ and $C$ a quantum circuit from $m + r$ to $m + r$ qubits for some $r \in \Nat\cup\{0\}$, and output set of $r$-to-$r$ qubit quantum channels. We say that $A$ is a unitary CTC generator if for every well-formed input $\left( m, C \right)$, if the pair is a pure PMG (as in Definition \ref{definition:pure_process_matrix_generator}) then $A$ outputs the channel $\pmo_{\left( m, C \right)}$ (and if the input is not a pure PMG, the output channel of $A$ can be any $r$-to-$r$ quantum channel).
\end{definition}

We can now define the complexity class $\BQPuctc$.

\begin{definition}[The Complexity Class $\BQPuctc$] \label{definition:BQPuctc}
    The complexity class $\BQPuctc$ is the set of promise problems $\prod:=\left( \YES,\NO \right)$ solvable by a quantum polynomial-time algorithm with access to unitary post-selection CTCs. Formally, there exists a quantum polynomial-time Turing machine $M$ such that, for every unitary CTC generator $A$ as in Definition \ref{definition:unitary_ctc_generator},
    $$
    \forall x\in\YES:
    \Pr\left[ M^A\left( x \right)=1 \right]\geq 1-2^{-|x|}
    \enspace ,
    $$
    and
    $$
    \forall x\in\NO:
    \Pr\left[ M^A\left( x \right)=0 \right]\geq 1-2^{-|x|}
    \enspace .
    $$
    The access and cost conventions are as follows:
    \begin{itemize}
        \item
        $M^A$ is the machine $M$ with oracle access to $A$ and the channels that it returns.

        \item
        Every generation query $\left( m,C \right)$ has a classical description. Queries may depend on the input $x$ and on the preceding computation, and may be adaptive.

        \item
        Every query $\left( m,C \right)$ to $A$ and every use of the channel it returns costs $\Theta\left( |C| \right)$ time.
    \end{itemize}
\end{definition}

\paragraph{Purifiable processes.}
Ara\'{u}jo, Feix, Navascu\'{e}s and Brukner \cite[Definition 3]{araujo2017purification} introduced purifiable processes and proposed purifiability as a necessary condition for physical realizability. Informally, a process is purifiable if it can be recovered from a larger pure process by initializing additional inputs to pure states and discarding additional outputs.

\begin{definition}[Purifiable Process Matrix] \label{definition:purifiable_process_matrix}
    Let $W$ be a process matrix as in Definition \ref{definition:process_matrix}, with feedback input $R$, feedback output $S$, external input $P$, and external output $F$, with a fixed partition of $R, S$ into feedback pairs. We say that $W$ is \emph{purifiable} if there are additional external registers $P', F'$ and a unitary
    $$
    U:\Hilb_R\otimes\Hilb_P\otimes\Hilb_{P'}
    \longrightarrow
    \Hilb_S\otimes\Hilb_F\otimes\Hilb_{F'}
    $$
    satisfying both of the following conditions:
    \begin{enumerate}
        \item
        $\widetilde W:=\Dket{U}\Dbra{U}$ is a process matrix for the same feedback pairs, with external input $P, P'$ and external output $F, F'$. Since $U$ is unitary, $\widetilde W$ is a pure process.

        \item
        Initializing $P'$ to zero and discarding $F'$ recovers $W$. Formally, writing
        $$
        V:=U\left( I_{\Hilb_R\otimes\Hilb_P}\otimes\ket{0}_{P'} \right)
        \enspace ,
        $$
        we have
        $$
        W=\tr_{F'}\left( \Dket{V}\Dbra{V} \right)
        \enspace .
        $$
    \end{enumerate}
    Validity of $\widetilde W$ is required on every input to $P'$, including inputs entangled with the other external registers, and under every choice of local quantum channels. No circuit-synthesis or efficiency requirement is imposed on $U$.
\end{definition}

\begin{definition}[Purifiable Process Matrix Generator] \label{definition:purifiable_process}
    Let $D=\left( n,C \right)$ be a PMG as in Definition \ref{definition:process_matrix_generator}. We say that $D$ is \emph{purifiable} if $\CJ\left( \channel_C \right)$ is a purifiable process matrix according to Definition \ref{definition:purifiable_process_matrix}, with the same $n$ separate single-qubit feedback pairs.

    Equivalently, there is a unitary $U$ defining a pure extension as above such that, for every matrix $X$ on the original circuit input,
    $$
    \channel_C\left( X \right)
    =
    \tr_{F'}\left[
    U\left( X\otimes\ket{0}\bra{0}_{P'} \right)U^\dagger
    \right]
    \enspace .
    $$
    If $U$ is implemented by a circuit $\widetilde C$, then $\left( n,\widetilde C \right)$ is a pure PMG. The existence of such a circuit over our fixed gate set, its size, and the complexity of constructing it are additional requirements, not part of purifiability.
\end{definition}

\section{Quantum Computation with Efficient General Process Matrices} \label{section:linear_CTC_computation}
We show that general process matrix computation can solve every problem in $\NP$. Our main tool turns an efficient preparation of an isometry's Choi state into a process matrix generator implementing that isometry. We first prove this construction and its purifiability, and then use it for unique-witness search, solving $\NP$, and coherent permutation inversion.

\subsection{Generating a unitary channel from a generator of its Choi state}\label{subsection:choi_to_isometry}
Preparing the Choi state of an isometry does not by itself give an efficient circuit implementing that isometry. The following lemma shows that it does give an efficient process matrix generator. We also prove a stronger validity guarantee that will allow us to establish purifiability.

\begin{lemma}[Generating an isometry from its Choi-state generator] \label{lemma:from_choi_to_isometry}
    Let $n,m\in\Nat\cup\{0\}$, and let $V\in\Lin\left(\states(n),\states(m)\right)$ be an isometry. Denote $d:=2^n$, and let
    $$
    \ket{J_V}
    :=
    \frac{1}{\sqrt d}\sum_{x\in\{0,1\}^n}\ket{x}\otimes V\ket{x}
    \enspace .
    $$
    Assume there is an ordinary polynomial-size quantum circuit $P_V$ that prepares the state $\ket{J_V}$ exactly. Then there is a classical deterministic polynomial-time algorithm that, given $P_V$, outputs a pair $(2n,C)$ which is a PMG (Definition \ref{definition:process_matrix_generator}) such that, for every input quantum state $\rho$,
    \begin{equation}
        \pmo_{(2n,C)}(\rho)=V\rho V^\dagger\enspace .\label{eq:linear-choi-generated}
    \end{equation}
    In particular, when $n=m$, the generated channel is unitary.

    \vspace{2mm}
    \noindent
    The pair $(2n,C)$ actually satisfies a stronger property than being a PMG. Specifically, the induced Process Matrix Operator is a quantum channel even when the $2n$ single-qubit feedback registers are treated as one joint register. That is, the validity condition holds for every quantum channel $\mathcal{A}$ acting jointly on all $2n$ feedback qubits, with arbitrary auxiliary input and output registers, rather than only for tuples of local quantum channels $\mathcal{A}_1,\ldots,\mathcal{A}_{2n}$ as in Definition \ref{definition:process_matrix}.
\end{lemma}
\begin{proof}
    Let $d:=2^n$, and denote by $P_j$, for $j\in\{0,1\}^{2n}$, the $n$-qubit Pauli matrices, chosen with phases so that $P_j^\dagger=P_j$, $P_j^2=I_{n}$, and $P_{0^{2n}}=I_{n}$. These matrices satisfy that the trace of their inner product with each other is always $0$ unless it is with themselves, specifically, $\tr(P_j^\dagger P_{j'})=d\cdot\mathbf{1}[j=j']$.

    Write
    $$
    \ket{\Phi_d}:=\frac{1}{\sqrt d}\sum_{x\in\{0,1\}^{n}}\ket{x,x}
    \quad \text{and} \quad
    \ket{\beta_j}:=(P_j\otimes I_{n})\ket{\Phi_d}\enspace .
    $$
    The vectors $\{\ket{\beta_j}\}_{j\in\{0,1\}^{2n}}$ form an orthonormal basis for $\states\left( 2n \right)$. Next, consider the $(2n + m)$-qubit state $\ket{\psi}\otimes\ket{J_V}$, and specifically, making a partial measurement to its leftmost $2n$ qubits in the basis $\{\ket{\beta_j}\}_{j\in\{0,1\}^{2n}}$.
    By definition, for every $j\in\{0,1\}^{2n}$, the measurement outputs $j$ with probability
    $$
    p_j
    :=
    \left\|
    \left(\bra{\beta_j}\otimes I_m\right)
    \left(\ket{\psi}\otimes\ket{J_V}\right)
    \right\|^2
    \enspace ,
    $$
    and conditioned on this outcome, the remaining $m$ qubits
    are in the state
    $$
    \frac{1}{\sqrt{p_j}}
    \left(\bra{\beta_j}\otimes I_m\right)
    \left(\ket{\psi}\otimes\ket{J_V}\right)
    \enspace ,
    $$
    whenever $p_j>0$. A useful observation is that for every quantum state $\ket{\psi}\in\states(n)$ and $j\in\{0,1\}^{2n}$ we have
    \begin{equation}
    \left( \bra{\beta_j} \otimes I_{m} \right) \left(\ket{\psi}\otimes\ket{J_V}\right)
    =\frac{1}{d} VP_j\ket{\psi}\enspace .\label{eq:linear-teleportation}
    \end{equation}
    To see the equality, expand the three vectors in the computational basis and use the definition of $\ket{\beta_j}$.
    
    Now, since $V^\dagger V=I$ and $P_j$ is unitary, the squared norm of the right-hand side is $1/d^2$, independently from $\ket{\psi}$. This altogether implies that measuring the first two registers in the basis $\{\ket{\beta_j}\}_j$ produces each $j$ with probability $1/d^2$, and the corresponding external state is $VP_j\ket{\psi}$. The corresponding statement for mixed inputs and inputs entangled with additional registers follows by linearity, where the transformation acts as the identity on the additional registers.

    \paragraph{The process matrix generator.}
    We describe a polynomial-size ordinary quantum circuit $C$ acting on $(2n + n)$-qubit inputs and outputting $2n + m$ qubits. Specifically, the feedback registers are the first $2n$ qubits, thus there are additional $n$ external input qubits and $m$ external output qubits. On an arbitrary input, $C$ executes the following procedure:
    \begin{enumerate}
        \item Measure the feedback input in the computational basis, and denote the outcome by $r\in\{0,1\}^{2n}$.
        \item Execute the circuit $P_V$ to prepare $\ket{J_V}$ in an $(n + m)$-qubits register $R_{V}$, and perform the measurement from Equation \eqref{eq:linear-teleportation} on the remaining $n$ qubits in the external input and the first $n$ qubits of $R_V$. Denote the measurement outcome by $j\in\{0,1\}^{2n}$.
        \item Set the feedback output to $a := r\oplus j$ by adding (mod 2) the measurement outcome to the feedback registers. Output the remaining $m$ qubits of $R_V$ as the external output, and discard the remaining registers.
    \end{enumerate}
    The circuit $C$ uses only ordinary quantum operations and note that its size is polynomial in $|P_V|+n+m$.

    \paragraph{Validity of the generator.}
    We next show that the pair $(2n,C)$ constructed above is a PMG. Recall that in our definition a PMG needs to remain a valid quantum channel even when arbitrary quantum channels, with arbitrary auxiliary registers, are inserted in the feedback. In fact, we prove a stronger property, where all the $2n$ feedback qubits are treated as a single register and the inserted quantum channel can act jointly on all of them.

    Let $E,E'$ be arbitrary auxiliary registers, and let $\mathcal A$ be an arbitrary quantum channel from the $2n$ feedback output qubits and $E$ to the $2n$ feedback input qubits and $E'$. Formally,
    $$
    \mathcal A:\Lin\left(\states(2n)\otimes\Hilb_E\right)
    \longrightarrow
    \Lin\left(\states(2n)\otimes\Hilb_{E'}\right)
    \enspace .
    $$
    Recall the matrix representation of a quantum channel from Definition \ref{definition:channel_matrices} through Kraus operators. Since $\mathcal A$ is a quantum channel, there exist matrices $A_1,\ldots,A_t\in\Lin\left(\states(2n)\otimes\Hilb_E,\states(2n)\otimes\Hilb_{E'}\right)$ such that for every input matrix $X$,
    $$
    \mathcal A(X)=\sum_{\ell=1}^{t}A_{\ell}XA_{\ell}^{\dagger}
    \quad \text{and} \quad
    \sum_{\ell=1}^{t}A_{\ell}^{\dagger}A_{\ell}=I_{\Lin\left(\states(2n)\otimes\Hilb_E\right)}
    \enspace .
    $$
    We use the Kraus operators $A_1,\ldots,A_t$ only for analyzing the quantum channel $\mathcal A$, and do not assume that the algorithm receives their descriptions or can compute them efficiently.

    Recall that the feedback channel $\mathcal A$ has input which is the out $a$ of $C$ on its first $2n$ qubits and output $r$ which is conceptually the input to $C$. For every $r,a\in\{0,1\}^{2n}$ and $\ell\in[t]$, denote by $A_{\ell}^{r,a}\in\Lin(\Hilb_E,\Hilb_{E'})$ the part of $A_{\ell}$ corresponding to the feedback input $a$ and the feedback output $r$. Formally, we define
    $$
    A_{\ell}^{r,a}
    :=
    \left(\bra{r}\otimes I_{\Lin(\Hilb_{E'})}\right)
    A_{\ell}
    \left(\ket{a}\otimes I_{\Lin(\Hilb_E)}\right)
    \enspace .
    $$
    Thus, on an auxiliary input state $\ket{\eta}\in\Hilb_E$, the vector $A_{\ell}^{r,a}\ket{\eta}$ is the $r$-component of $A_{\ell}\left(\ket{a}\otimes\ket{\eta}\right)$ after removing the feedback output register. In particular, the second equality in the matrix representation of $\mathcal A$ gives, for every fixed $a\in\{0,1\}^{2n}$,
    \begin{equation}
    \begin{aligned}
    \sum_{r\in\{0,1\}^{2n}}\sum_{\ell=1}^{t}\left(A_{\ell}^{r,a}\right)^\dagger A_{\ell}^{r,a}
    &=\left(\bra{a}\otimes I_{\Lin(\Hilb_E)}\right)
    \left(\sum_{\ell=1}^{t}A_{\ell}^{\dagger}A_{\ell}\right)
    \left(\ket{a}\otimes I_{\Lin(\Hilb_E)}\right)\\
    &=I_{\Lin(\Hilb_E)}
    \enspace .
    \end{aligned}\label{eq:linear-local-tp}
    \end{equation}

    We now consider the transformation obtained by connecting the circuit $C$ through the channel $\mathcal A$. Recall the execution of $C$: it first measures the feedback input and obtains $r$, then makes the Bell-basis measurement and obtains $j$, and finally sets the feedback output to $a=r\oplus j$. By Equation \eqref{eq:linear-teleportation}, the operation applied to the external input in the branch with outcome $j$ is $VP_j/d$. For this branch to contribute to the connected transformation, the channel $\mathcal A$ must take the feedback output $a=r\oplus j$ back to the feedback input $r$. According to the above notation, the corresponding operation on the auxiliary registers is $A_{\ell}^{r,r\oplus j}$, for some $\ell\in[t]$.

    Consequently, the connected transformation has the following matrix representation. For $r,j\in\{0,1\}^{2n}$ and $\ell\in[t]$, define the matrices
    $$
    K_{r,j,\ell}:=\frac1d VP_j\otimes A_{\ell}^{r,r\oplus j}
    \enspace .
    $$
    Since the feedback input and the Bell basis are measured during the execution of $C$, the different outcomes $r,j$ contribute as separate branches. Therefore, denoting by $\mathcal G_{\mathcal A}$ the induced Process Matrix Operator, for every matrix $X$ on the external input and the auxiliary input $E$ we have
    $$
    \mathcal G_{\mathcal A}(X)
    =\sum_{r,j\in\{0,1\}^{2n}}\sum_{\ell=1}^{t}
    K_{r,j,\ell}XK_{r,j,\ell}^{\dagger}
    \enspace .
    $$
    In particular, this shows that $\mathcal G_{\mathcal A}$ is completely positive, also when its input is entangled with arbitrary additional registers. It remains to prove that the induced operator is trace-preserving.

    Recall that $V$ is an isometry and the matrices $P_j$ are unitary. Thus $V^\dagger V=P_j^\dagger P_j=I_{n}$, and we have
    \begin{align*}
    \sum_{r,j,\ell}K_{r,j,\ell}^{\dagger}K_{r,j,\ell}
    &=\frac1{d^2}I_{\Lin(\states(n))}\otimes
    \sum_{r,j,\ell}\left(A_{\ell}^{r,r\oplus j}\right)^\dagger A_{\ell}^{r,r\oplus j}\\
    &=\frac1{d^2}I_{\Lin(\states(n))}\otimes
    \sum_{a\in\{0,1\}^{2n}}\sum_{r,\ell}\left(A_{\ell}^{r,a}\right)^\dagger A_{\ell}^{r,a}\\
    &=\frac1{d^2}I_{\Lin(\states(n))}\otimes
    \sum_{a\in\{0,1\}^{2n}}I_{\Lin(\Hilb_E)}\\
    &=I_{\Lin(\states(n)\otimes\Hilb_E)}
    \enspace .
    \end{align*}
    In the second equality we simply used the change of variables $a=r\oplus j$, and the third equality follows from Equation \eqref{eq:linear-local-tp}. The last equality holds because there are $2^{2n}=d^2$ possible strings $a$. By Definition \ref{definition:channel_matrices}, this proves that the induced operator is trace-preserving. Combining this with complete positivity, we get that the Process Matrix Operator induced by $C$ is a quantum channel for every $\mathcal A$ and arbitrary auxiliary registers $E,E'$. Thus, $(2n,C)$ is a PMG, and the stronger validity property in the statement of the lemma also holds.

    \paragraph{The generated channel.}
    We are left to show the action of the generated channel $\pmo_{(2n,C)}$, which is obtained by inserting the identity channel on the feedback registers without auxiliary input. In this case the inserted channel has the single matrix $A_1=I_{\Lin(\states(2n))}$, and therefore $A_1^{r,a}=\mathbf{1}[r=a]$. Recall that the feedback output of $C$ is $a=r\oplus j$, so the condition $r=a$ holds if and only if $j=0^{2n}$. Thus, in the matrix representation of the derived channel, the only nonzero operators are
    $$
    K_{r,0^{2n},1}=\frac1d V
    \enspace ,
    $$
    one for every $r\in\{0,1\}^{2n}$. In particular, recalling that there are $d^2$ different values of $r$, for every quantum state $\rho$ we have
    $$
    \pmo_{(2n,C)}(\rho)
    =\sum_{r\in\{0,1\}^{2n}}\frac1{d^2}V\rho V^\dagger
    =V\rho V^\dagger
    \enspace .
    $$
    This completes the proof.
\end{proof}

\paragraph{Purifiability of the construction.}
The stronger validity property in Lemma \ref{lemma:from_choi_to_isometry} also implies purifiability. We state the implication separately, since it applies beyond our particular construction.

\begin{proposition}[Joint validity implies purifiability]
\label{proposition:joint_validity_implies_purifiability}
    Let $D=(n,C)$ be a PMG. Suppose that its induced transformation remains a quantum channel when an arbitrary quantum channel $\mathcal A$ acts jointly on all $n$ feedback qubits, including arbitrary auxiliary input and output registers. Then $D$ is purifiable in the sense of Definition \ref{definition:purifiable_process}.
\end{proposition}

\begin{proof}
    Write $R,S$ for the feedback input and output of $C$, and $P,F$ for its external input and output. Treat all feedback pairs as one pair, so that $\mathcal A$ acts from $S$ to $R$, possibly together with auxiliary registers. Denote the resulting transformation by $\mathcal G_{\mathcal A}$.

    The mapping $\mathcal A\mapsto\mathcal G_{\mathcal A}$ is a deterministic quantum supermap in the terminology of \cite{chiribella2008transforming}. Indeed, it is linear and completely positive in the CJ representation, since its Choi operator is $\CJ(\channel_C)$ up to a reordering of registers. By our assumption, it maps every quantum channel to a quantum channel, including when auxiliary registers are present.

    Theorem 1 of \cite{chiribella2008transforming} therefore gives a memory register $M$, an additional output register $G$, and isometries
    $$
    B_1:\Hilb_P\longrightarrow\Hilb_S\otimes\Hilb_M
    \quad \text{and} \quad
    B_2:\Hilb_R\otimes\Hilb_M\longrightarrow\Hilb_F\otimes\Hilb_G
    $$
    such that, for every $\mathcal A$ without auxiliary registers,
    $$
    \mathcal G_{\mathcal A}(X)
    =
    \tr_G\left[
    B_2\,
    (\mathcal A\otimes\id_M)\left(B_1XB_1^\dagger\right)
    B_2^\dagger
    \right]
    \enspace .
    $$
    The same identity holds with arbitrary auxiliary registers, on which the surrounding operations act as the identity. Crucially, $B_1$ and $B_2$ are fixed independently of $\mathcal A$.

    Extend $B_1$ and $B_2$ to unitaries by adding input registers $K_1,K_2$, respectively. Denote these unitaries by
    $$
    U_1:\Hilb_P\otimes\Hilb_{K_1}
    \longrightarrow\Hilb_S\otimes\Hilb_M
    \quad \text{and} \quad
    U_2:\Hilb_R\otimes\Hilb_M\otimes\Hilb_{K_2}
    \longrightarrow\Hilb_F\otimes\Hilb_G
    \enspace .
    $$
    Memory and environment spaces may be enlarged to qubit registers as needed. Initializing $K_i$ to zero recovers $B_i$. Define a unitary $U$ from $RPK_1K_2$ to $SFG$ by applying $U_1$ to $PK_1$, and then applying $U_2$ to $RMK_2$, leaving $S$ untouched.

    To verify process validity, insert an arbitrary joint channel $\mathcal A$ from $SE$ to $RE'$. The connected transformation is an ordinary circuit: apply $U_1$, pass $S$ through $\mathcal A$, and apply $U_2$ to the returned register $R$ and the kept registers $M,K_2$. Every operation is a quantum channel, so the connected transformation is a quantum channel on every joint input to $P,K_1,K_2,E$. If $\mathcal A$ is unitary, every operation is unitary. In particular, this holds for every tuple of local unitary operations on the original separate feedback pairs. Thus, $\CJ(\Ad_U)$ is a pure process.

    Finally, initialize $K_1,K_2$ to zero and discard $G$. By the realization above, this recovers $\mathcal G_{\mathcal A}$ for every $\mathcal A$, including its auxiliary extensions. This also recovers the entire original generating channel $\channel_C$: choose $\mathcal A$ to swap $S$ with an external $n$-qubit input, returning that input to $R$ and sending $S$ to an external output. With this choice, the connected transformation is precisely $\channel_C$, up to register ordering. Consequently, for every matrix $X$ on $RP$,
    $$
    \channel_C(X)
    =
    \tr_G\left[
    U\left(X\otimes\ket{0}\bra{0}_{K_1K_2}\right)U^\dagger
    \right]
    \enspace .
    $$
    Taking $P'=K_1K_2$ and $F'=G$ proves purifiability.
\end{proof}

Applying Proposition \ref{proposition:joint_validity_implies_purifiability} to the stronger validity guarantee in Lemma \ref{lemma:from_choi_to_isometry}, every process obtained by that lemma is purifiable. In particular, this applies to the unique-witness and permutation-inversion constructions below. The argument establishes existence of a pure extension; it does not provide a polynomial-time algorithm for constructing its description from the supplied Choi-state preparation circuit.

\subsection{Finding a unique satisfying input}\label{subsection:unique_witness_search}

Our next construction applies the previous lemma to the problem of finding the unique satisfying input of an efficiently computable Boolean function.

\begin{theorem}[Unique-witness search]\label{theorem:unique_witness_search}
    There exists a classical deterministic polynomial-time algorithm $M$ that, given a Boolean circuit $v:\{0,1\}^k\rightarrow\{0,1\}$ with exactly one satisfying input $z$, computes a pair $(n,C)$ for $n = O\left( k \right)$ which is a PMG and whose generated channel is unitary. Moreover, using one query to this generated channel, $M$ recovers $z$ with probability $1$.
\end{theorem}
\begin{proof}
    Let $v$ be a Boolean circuit as in the statement, with unique satisfying input $z$. Introduce one additional input bit and require it to be $1$ for verification to pass. We continue to write $v,k,z$ for the padded verifier, its input length, and its unique witness, respectively. Thus, $z\neq0^k$. We discard the added bit when reporting the witness for the original instance. Denote $d:=2^k$. Consider the following unitary matrix $S_z\in\bbC^{d\times d}$:
    $$
    S_z\ket{x}:=
    \begin{cases}
    \ket{z},&x=0^k,\\
    \ket{0^k},&x=z,\\
    \ket{x},&\text{otherwise}.
    \end{cases}
    $$
    Observe that $S_z$ is unitary, as it is a permutation of the computational-basis vectors, and that $S_z\ket{0^k}=\ket{z}$. We will show how to generate the channel corresponding to $S_z$ without first computing $z$.

    By \lemref{lemma:from_choi_to_isometry}, it suffices to give an ordinary polynomial-size quantum circuit that prepares the state
    \begin{equation}
    \ket{J_z}
    :=\frac1{\sqrt d}
    \left(
    \sum_{x\in\{0,1\}^k\setminus\{0^k,z\}}\ket{x,x}
    +\ket{0^k,z}+\ket{z,0^k}
    \right)
    \enspace .\label{eq:linear-unique-choi}
    \end{equation}
    We describe such a circuit next.

    \paragraph{Preparing the state $\ket{J_z}$.}
    Start with the uniform superposition on $k+1$ qubits,
    $$
    \frac1{\sqrt{2d}}\sum_{x\in\{0,1\}^k}\sum_{b\in\{0,1\}} \ket{x,b}
    \enspace .
    $$
    Add registers $A,C,G$ on $k$, $k$ and $1$ qubits respectively, initially in the all-zero state $\ket{0^{2k+1}}$. By ordinary reversible computation, evaluate $v(x)$ and compute the following (additions are modulo 2):
    \begin{enumerate}
        \item If $x\neq0^k$, $v(x)=0$, and $b=0$, add $(x,x)$ to $(A, C)$ and $1$ to $G$.
        \item If $v(x)=1$ and $b=0$, add $(0^k,x)$ to $(A, C)$ and $1$ to $G$.
        \item If $v(x)=1$ and $b=1$, add $(x,0^k)$ to $(A, C)$ and $1$ to $G$.
        \item In all other cases, leave $(A, C, G)$ unchanged.
    \end{enumerate}
    Observe that, on all branches with $G=1$, the pair $(A,C)$ uniquely determines the original input $(x,b)$: if $A=0^k$ then $(x,b)=(C,0)$; if $C=0^k$ then $(x,b)=(A,1)$; and otherwise $(x,b)=(A,0)$. Therefore, conditioned on $G=1$, the original registers $x,b$ and all computational work registers can be reversibly returned to the all-zero state, without changing $(A,C)$.

    There are exactly $d$ successful pairs $(x,b)$: one for each $x\notin\{0^k,z\}$ and two more for $x=z$. All successful pairs give different values of $(A,C)$, and these values are exactly the $d$ pairs appearing in Equation \eqref{eq:linear-unique-choi}. Therefore, the probability that $G=1$ is exactly $1/2$, and conditioned on this event the output registers are in state $\ket{J_z}$, with no entangled information left behind.

    Recall that in our Lemma \ref{lemma:from_choi_to_isometry} we use a circuit that makes $\ket{J_{z}}$ with probability $1$, and the above approach achieves probability $1/2$, and can get arbitrarily close to $1$ by repetition, but not exactly. We next turn this constant-probability preparation to be exact.
    
    Append an independent Hadamard bit $\ket{+}_T$ in a single-qubit register $T$, and call a branch good if $G=T=1$. The probability of this event is exactly $1/4$. Let $N$ be the total number of qubits used by this preparation, including all work registers. Denote the unitary computation described so far by $A$, let $\ket{\Psi}:=A\ket{0^N}$, and let $\Pi$ be the projector onto the good subspace. Then we can write
    $$
    \ket{\Psi} := \frac{1}{2}\ket{S_0} + \frac{\sqrt{3}}{2}\ket{S_1}
    \enspace ,
    $$
    for normalized orthogonal states $S_0,S_1$, where $\ket{S_0}$ has the output registers $(A, C)$ in state $\ket{J_z}$ and every work register in a fixed (untangled) state. Direct calculation gives
    \begin{equation}
    \left(2\ket{\Psi}\bra{\Psi}-I\right)
    (I-2\Pi)\ket{\Psi}=\ket{S_0}\enspace .\label{eq:linear-unique-amplification}
    \end{equation}
    Indeed, the right reflection produces $-\ket{S_0}/2+\sqrt3\ket{S_1}/2$, whose inner product with $\ket{\Psi}$ is $1/2$. Reflecting this state about $\ket{\Psi}$ (which is the second operation) therefore gives $\ket{S_0}$. The first (rightmost) operation in Equation \eqref{eq:linear-unique-amplification} is an ordinary phase flip controlled on $G=T=1$ (i.e., apply a control $Z$ two-qubit gate with control register $G$ and target register $T$) and the second 
    reflection is implemented by applying $A^\dagger$, then $2\ket{0^N}\bra{0^N}-I$, and finally $A$. Both are implementable by polynomial-size quantum circuits. Thus we obtain, without measurement or postselection, an exact polynomial-size preparation of $\ket{J_z}$.

    \paragraph{Process generation and correctness.}
    Apply \lemref{lemma:from_choi_to_isometry} to this preparation circuit, with $V:=S_z$. We obtain a polynomial-size ordinary quantum circuit $C$ for which $(n,C) := (2k,C)$ is a PMG and whose generated channel is
    $$
    \pmo_{(n,C)}(\rho)=S_z\rho S_z^\dagger
    \enspace .
    $$
    The generated channel is unitary, even though the original circuit $C$ may use measurements and discarded registers. By applying this generated channel to $\ket{0^k}\bra{0^k}$ and measuring the output, the algorithm recovers the unique satisfying input $z$ with probability $1$. The construction and every oracle use cost polynomial resources, which concludes the proof.
\end{proof}

\subsection{Solving \texorpdfstring{$\NP$}{NP}}\label{sec:np}

We now use \thmref{theorem:unique_witness_search} and the Valiant--Vazirani reduction from Theorem \ref{theorem:valiant_vazirani} to solve $\SAT$, and consequently all of $\NP$, in $\BQPlctc$. We apply our unique-witness construction to each reduced formula and verify the returned candidate assignment before accepting.

\begin{theorem}\label{thm:np}
    $$
    \NP\subseteq\BQPlctc
    \enspace .
    $$
    The algorithm makes polynomially many nonadaptive generation queries. Whenever the formula used in a query has exactly one satisfying assignment, the query is a PMG and its generated channel is unitary.
\end{theorem}
\begin{proof}
    It suffices to solve $\SAT$. Let $A$ be any generation oracle satisfying Definition \ref{definition:linear_ctc_generator}, and let $\varphi$ an input CNF formula with $m$ input variables. If the input does not encode a CNF formula, reject. Formulas with no variables are evaluated directly. We describe a quantum polynomial-time Turing Machine $M$ that, with access to $A$, decides whether $\varphi\in\SAT$.

    Let $T:=\left\lceil K|\varphi|^2\right\rceil$, for a sufficiently large absolute constant $K$. The machine $M$ executes the following procedure $T$ times, with independent randomness for each execution of the Valiant-Vazirani reduction:
    \begin{enumerate}
        \item Execute $T_{VV}(\varphi)$ from Theorem \ref{theorem:valiant_vazirani}, obtaining a CNF formula $\varphi^*$.
        \item Apply the deterministic circuit construction from \thmref{theorem:unique_witness_search} to the Boolean circuit that verifies a satisfying assignment of $\varphi^*$. Denote the resulting pair by $(n,C)$, and query $A(n,C)$.
        \item Apply the returned channel to the designated all-zero input, and measure its output in the computational basis. Remove the extra witness bit introduced in the proof of \thmref{theorem:unique_witness_search}, obtaining a candidate assignment $w$. Accept if $\varphi^*(w)=1$.
    \end{enumerate}
    If none of the iterations accepted, reject.

    The construction in Step 2 outputs a well-formed circuit description for every $\varphi^*$. If $\varphi^*$ has exactly one satisfying assignment, \thmref{theorem:unique_witness_search} guarantees that this description is a PMG whose generated channel finds the assignment.

    \paragraph{Correctness.}
    Suppose first that $\varphi\notin\SAT$. By Theorem \ref{theorem:valiant_vazirani}, every formula $\varphi^*$ generated in Step 1 is unsatisfiable. Therefore, no candidate assignment can pass Step 3, regardless of which channels $A$ returns. The machine rejects with probability $1$.

    Now suppose that $\varphi\in\SAT$. There is a constant $c>0$ such that each execution of Step 1 produces a uniquely satisfiable formula with probability at least $c/m$. Whenever this event occurs, \thmref{theorem:unique_witness_search} guarantees that Step 2 supplies the channel that recovers the unique satisfying assignment. Step 3 therefore accepts with probability $1$. Consequently, the probability that all $T$ iterations reject is at most
    $$
    \left(1-\frac{c}{m}\right)^T
    \leq
    \exp\left(-\frac{cT}{m}\right)
    \enspace .
    $$
    Since $m\leq|\varphi|$, choosing $K$ sufficiently large gives
    $$
    \exp\left( -\frac{cT}{m} \right)
    \leq
    \exp\left( -cK|\varphi| \right)
    \leq
    2^{-|\varphi|}
    \enspace .
    $$

    \paragraph{Running time and generation queries.}
    The reduction, the circuit construction, and the final verification all take polynomial time. Every generation query and every use of its returned channel have polynomial cost under Definition \ref{definition:BQPlctc}. As a side note, one can also verify that the generation queries are nonadaptive. This proves $\SAT\in\BQPlctc$.
\end{proof}

\subsection{Inverting an efficiently computable permutation}\label{sec:permutations}
We end with an application of \lemref{lemma:from_choi_to_isometry} to computing a function, rather than finding a single witness. The construction uses the same computational model $\BQPlctc$, and provides a unitary that inverts a permutation on arbitrary quantum inputs.

\begin{theorem}[Permutation inversion]\label{theorem:permutation_inversion}
    There exists a deterministic polynomial-time algorithm that, given an ordinary classical circuit for a permutation $\pi:\{0,1\}^n\rightarrow\{0,1\}^n$, outputs a PMG whose generated channel is the unitary transformation
    \begin{equation}
    \ket{y}\longmapsto\ket{\pi^{-1}(y)}\enspace .\label{eq:linear-permutation-inverse}
    \end{equation}
    In particular, this transformation acts coherently on arbitrary input superpositions and inputs entangled with additional registers.
\end{theorem}
\begin{proof}
    Let $\pi$ be a permutation with polynomial-size evaluation circuit. The strategy is to use Lemma \ref{lemma:from_choi_to_isometry} by a simply circuit that generates the Choi state of the inverse $\pi^{-1}$.
    
    Prepare the uniform superposition over $x$, compute $\pi(x)$ into a fresh register and SWAP the two registers. Since $\pi$ is a permutation, the same state equals
    $$
    \frac1{\sqrt{2^n}}\sum_{x\in\{0,1\}^n}\ket{\pi(x)}\otimes\ket{x}
    =
    \frac1{\sqrt{2^n}}\sum_{y\in\{0,1\}^n}\ket{y}\otimes\ket{\pi^{-1}(y)}
    =
    \ket{J_{U_{\pi^{-1}}}}
    \enspace ,
    $$
    where $U_{\pi^{-1}}\ket{y}:=\ket{\pi^{-1}(y)}$ is unitary. Applying \lemref{lemma:from_choi_to_isometry} gives a polynomial-size PMG with generated channel $U_{\pi^{-1}}$. The computational model therefore gives access to this channel. This proves the theorem.
\end{proof}

\section{Quantum Computation with Efficient and Pure Process Matrices} \label{section:unitary_CTC_computation}
In this section we construct pure processes, with unitary generating circuits constructed efficiently from the input. We first prove a general result on quantum changes of basis and then use it to invert USO outmaps. This gives our algorithm for $\USO$ and its consequences through classical reductions. We end by discussing connections to physical implementation.

\subsection{A unitary generator for the USO Oracle} \label{subsection:uso_generator}
We recall the unique-sink orientation problem, introduced in \cite{szabo2001unique}. The input describes an exponentially large graph by a circuit which computes the directions of the edges incident to a given vertex.

\begin{definition} [Unique-Sink Orientation (USO) Problem] \label{definition:uso}
    Let $n\in\Nat$. The $n$-dimensional Boolean cube has vertex set $\{0,1\}^n$, where two vertices are adjacent if they differ in exactly one coordinate. For every $i\in[n]$, let $e_i\in\zo^n$ be the string whose $i$-th coordinate is $1$ and all other coordinates are $0$. An orientation of the cube assigns a direction to every edge, and is represented by a function $s:\{0,1\}^n\to\{0,1\}^n$, where $s_i(x)=1$ if the edge between $x$ and $x\oplus e_i$ is directed away from $x$, and $s_i(x)=0$ if it is directed to $x$. In particular, $s_i(x)\neq s_i(x\oplus e_i)$ for every $x$ and $i$. We call $s$ the outmap of the orientation.

    A face of the cube is obtained by fixing an arbitrary subset $F\subseteq[n]$ of coordinates to specified values. A sink of a face is a vertex of that face with no outgoing edges along any of its free coordinates. An orientation is called a unique-sink orientation (USO) if every face has exactly one sink. This includes the entire cube, obtained by fixing no coordinates, and faces consisting of a single vertex, obtained by fixing all coordinates (which trivially satisfies the requirement, because there are no free coordinates in that case).

    In the computational problem $\USO$, the input is a classical Boolean circuit computing an outmap $s$, promised that the corresponding orientation is a USO. The output is the unique vertex $z\in\{0,1\}^n$ satisfying $s(z)=0^n$. Running times are measured in the description length of the input circuit.
\end{definition}

The following characterization of unique-sink orientations is given in \cite[Lemma 2.3]{szabo2001unique}. We also state an equivalent formulation that will be useful for our constructions. For completeness, we provide a proof in Appendix \ref{app:uso-facts}.\footnote{The original paper does not provide a proof, but this combinatorial fact is well-known and folklore.}

\begin{lemma} [Pairwise characterization of USO] \label{lem:uso-characterization}
    A function $s:\{0,1\}^n\to\{0,1\}^n$ is a USO outmap if and only if
    \begin{equation}
    x\neq x'\quad\Longrightarrow\quad\exists i\in[n]:\ x_i\neq x_i'\ \text{ and }\ s_i(x)\neq s_i(x')
    \enspace .\label{equation:uso_characterization}
    \end{equation}
    In particular, every USO outmap is a permutation of $\{0,1\}^n$. For $h(x):=x\oplus s(x)$, Equation \eqref{equation:uso_characterization} is equivalent to
    \begin{equation}
    x\neq x'\quad\Longrightarrow\quad\exists i\in[n]:\ x_i\neq x_i'\ \text{ and }\ h_i(x)=h_i(x')
    \enspace .\label{equation:uso_combinatorial_condition}
    \end{equation}
\end{lemma}

\paragraph{Main technical theorem.}
We next construct a pure PMG for a family of quantum changes of basis. For each computational-basis label $\ket{x}$, consider a product state obtained by applying a single-qubit unitary at each coordinate, choosing between two gates according to $h\left( x \right)$. The generated transformation will map this state back to $\ket{x}$. The condition in Equation \eqref{equation:uso_combinatorial_condition} guarantees both that these states form an orthonormal basis and that our generating circuit describes a pure process. Our algorithm for $\USO$ will follow as a special case.

\begin{theorem} [Generating a USO-decomposable unitary by a pure PMG] \label{theorem:decomposable_unitary_generation}
    Let $h:\{0,1\}^n\to\{0,1\}^n$ be given by a classical Boolean circuit satisfying Equation \eqref{equation:uso_combinatorial_condition}. Let $\{ U_{i,a} \}_{ i \in [n], a \in \zo }$ a set of single-qubit unitaries, each given by an ordinary unitary circuit. For every $x \in \zo^n$ define
    $$
    U(x):=
    \bigotimes_{i=1}^{n}U_{i,h_i(x)}
    \quad\text{and}\quad
    \ket{b_x}:=U(x)^\dagger\ket{x}
    \enspace .
    $$
    Then, the following holds.
    \begin{itemize}
        \item
        The vectors $\{\ket{b_x}\}_{x\in\{0,1\}^n}$ form an orthonormal basis of $\states(n)$ and thus the matrix
        \begin{equation}
        T:=\sum_{x\in\{0,1\}^n}\ket{x}\bra{b_x}
        =\sum_{x\in\{0,1\}^n}\ket{x}\bra{x}U(x)
        \enspace
        \label{eq:basis-map}
        \end{equation}
        is unitary.

        \item 
        Moreover, there is a classical deterministic polynomial-time algorithm that, given the circuits $\{ U_{i,a} \}_{ i \in [n], a \in \zo }$ and $h$, outputs a pair $(n,C)$ which is a pure PMG and whose generated Process Matrix Operator unitary (as in Definition \ref{definition:pure_process_matrix_generator}) is the matrix $T$.  
    \end{itemize}
\end{theorem}

\begin{proof}
    We first verify the claim about the basis. Every $\ket{b_x}$ has norm one because $U(x)$ is unitary. For different $x,x'$, choose $i$ as in Equation \eqref{equation:uso_combinatorial_condition}. In the inner product $\braket{b_x}{b_{x'}}$, the factor at this coordinate is
    $$
    \bra{x_i}U_{i,h_i(x)}U_{i,h_i(x')}^\dagger\ket{x_i'}
    =\braket{x_i}{x_i'}=0
    \enspace .
    $$
    Thus the vectors are orthonormal, and there are $2^n$ of them in $\states(n)$.

    \paragraph{The process matrix generator.}
    We describe a polynomial-size unitary quantum circuit $C$ acting on $2$ registers $R, Q$, each having $n$ qubits.
    Specifically, let $R=R_1\cdots R_n$ be the $n$-qubit feedback input register with one qubit in each $R_i$, and let $Q$ be the $n$-qubit external input register. On a basis input $\ket{x}_R$ and an arbitrary quantum state $\ket{\psi}_Q$, the circuit $C$ performs the following operations:
    \begin{enumerate}
        \item Compute $h(x)$ in an initially all-zero, $n$-qubit work register $W$.
        \item For every $i\in[n]$, apply $U_{i,h_i(x)}$ to qubit $i$ of $Q$, controlled on the corresponding bit of $W$.
        \item Reverse the computation of $h(x)$, returning $W$ to its all-zero state. The register $R$ has not been changed, so this uncomputation does not require knowing the contents of $Q$.
        \item Swap $R$ and $Q$.
    \end{enumerate}
    Consequently, the channel of $C$ is the unitary $U_C$ given by
    \begin{equation}
    U_C\left(\ket{x}_R\otimes\ket{\psi}_Q\right)
    =
    \left(U(x)\ket{\psi}\right)_R\otimes\ket{x}_Q
    \enspace .\label{eq:basis-C}
    \end{equation}

    \paragraph{Validity of the generator.}
    We prove the equivalent condition in Definition \ref{definition:pure_process_matrix_generator}. Let $E_1,\ldots,E_n$ be arbitrary additional registers, and let $V_i$ be an arbitrary unitary on $R_iE_i$. Write $E := E_1, \cdots, E_n$. Recall again that these registers and unitaries are used only in the proof of validity, the algorithm is not given their descriptions and need not implement them. Our goal is to show that the matrix
    $$
    T_{\mathbf V} :=
    \tr_R\left[
    \left(\bigotimes_{i=1}^n V_i\right)
    \left(U_C\otimes I_{\Lin(\Hilb_E)}\right)
    \right]
    $$
    is a unitary transformation on the registers $Q, E$.

    Our first step will be to look on the entire transformation before applying the partial trace. Let $r\in\{0,1\}^n$, $y\in\{0,1\}^n$ be inputs to the $R$ and $Q$ registers respectively and for $i \in [n]$ let $w_i \in \zo^{|E_{i}|}$ be an input to the $E$ register. We have the equalities
    \begin{equation}
    \begin{aligned} \label{equation:simplify_pre_trace_unitary}
    \left(\bigotimes_{i=1}^n V_i\right)
    &\left(U_C\otimes I_{\Lin(\Hilb_E)}\right)
    \left(\ket{r}_R\otimes\ket{y}_Q\otimes\ket{w}_E\right)
    \\
    &\underset{\text{Equation } \eqref{eq:basis-C}}{=}
    \left(\bigotimes_{i=1}^n V_i\right)
    \left(
    \left( U(r)\ket{y} \right)_R \otimes \ket{r}_Q \otimes\ket{w}_E
    \right)
    \\
    &\underset{\text{Def. of }U(r)}{=}
    \left(\bigotimes_{i=1}^n V_i\right)
    \left(
    \left( \bigotimes_{i=1}^{n}U_{i,h_i(r)}\ket{y_i} \right)_R \otimes \ket{r}_Q \otimes\ket{w}_E
    \right)
    \enspace .
    \end{aligned}
    \end{equation}
    We next want to trace out the $R$ register from the above transformation.

    \vspace{2mm}
    \noindent
    \textit{Writing the state as a standard-basis superposition over the feedback register.}
    To compute the partial trace explicitly, we first expand the feedback output in the computational basis. This will let us identify the component which stays in each summand of the partial trace. 
    Let us focus on registers $R$ and $E$ in Equation \ref{equation:simplify_pre_trace_unitary} and stop thinking about this as the output of a computation, and just as a state -- we seek to write this state in a way that the $R$ register contains a superposition of computational-basis vectors. Even more specifically, the state is enumerated by $r \in \zo^n$ which we can fix, so what remains in registers $R, E$ is
    $$
    \bigotimes_{i=1}^n V_i \cdot \left( U_{i,h_i(r)}\ket{y_i}_{R_i} \otimes \ket{w_i}_{E_{i}} \right) \enspace ,
    $$
    where for each $i \in [n]$, we define $w_{i} \in \zo^{|E_{i}|}$ as the disjoint $i$-th part of the string $w \in \zo^{|E|}$.
    
    Next, for each value $i \in [n]$, $a\in\{0,1\}$, define the unitary
    $$
    B_{i,a}:=V_i\left(U_{i,a}\otimes I_{\Lin(\Hilb_{E_i})}\right)
    \enspace .
    $$
    For bits $u,v\in\{0,1\}$, let $B_{i,a}^{u,v}$ be the matrix on $E_i$ obtained by fixing the input feedback bit of $B_{i,a}$ to $v$ and its output feedback bit to $u$. Formally,
    $$
    B_{i,a}^{u,v}:=
    \left( \bra{u} \otimes I_{\Lin(\Hilb_{E_i})} \right)
    B_{i,a}
    \left(\ket{v}\otimes I_{\Lin(\Hilb_{E_i})}\right)
    \enspace .
    $$
    Equivalently, $B_{i,a}=\sum_{u,v}\ket{u}\bra{v}\otimes B_{i,a}^{u,v}$. We deduce the equality:
    $$
    \bigotimes_{i=1}^n V_i \cdot \left( U_{i,h_i(r)}\ket{y_i}_{R_i} \otimes \ket{w_i}_{E_{i}} \right)
    =
    \sum_{z\in\{0,1\}^n}
    \ket{z}_R \otimes
    \left(\bigotimes_{i=1}^n B_{i,h_i(r)}^{z_i,y_i}\right)\ket{w}_E
    \enspace ,
    $$
    which overall, given Equation \ref{equation:simplify_pre_trace_unitary} implies
    \begin{align*}
    \left(\bigotimes_{i=1}^n V_i\right)
    \left(U_C\otimes I_{\Lin(\Hilb_E)}\right)
    &\left(\ket{r}_R\otimes\ket{y}_Q\otimes\ket{\eta}_E\right)
    \\&=
    \sum_{z\in\{0,1\}^n}
    \ket{z}_R\otimes\ket{r}_Q\otimes
    \left(\bigotimes_{i=1}^n B_{i,h_i(r)}^{z_i,y_i}\right)\ket{\eta}_E
    \enspace .
    \end{align*}
    
    Therefore, when we take the partial trace, only the term $z=r$ survives. The external output $\ket{r}_Q$ remains, together with the matrix on $E$ multiplying that term. Summing over all feedback input labels $r$, we obtain
    \begin{equation} \label{equation:T_V_descripion}
    T_{\mathbf V}\left(\ket{y}_Q\otimes\ket{\eta}_E\right)
    =
    \sum_{r\in\{0,1\}^n}
    \ket{r}_Q\otimes
    \left(\bigotimes_{i=1}^n B_{i,h_i(r)}^{r_i,y_i}\right)\ket{\eta}_E
    \enspace .
    \end{equation}

    \vspace{2mm}
    \noindent
    \textit{Proving unitarity by block-wise analysis.}
    We would next like to prove that $T_{\mathbf V} \cdot T_{\mathbf V}^\dagger = I_{\Lin(\Hilb_Q\otimes\Hilb_E)}$. One view that Equation \ref{equation:T_V_descripion} allows is that the input and output of the $Q$ register completely determines the transformation applied to the remaining system $E$. Accordingly we will analyze $T_{\mathbf V} \cdot T_{\mathbf V}^\dagger$ block-wise, partitioned into blocks fixing the input-output of $Q$. 

    For $x,y\in\{0,1\}^n$, define the matrix $(T_{\mathbf V})_{x,y}\in\Lin(\Hilb_E)$ by
    $$
    (T_{\mathbf V})_{x,y}
    :=
    \left(\bra{x}_Q\otimes I_{\Lin(\Hilb_E)}\right)
    T_{\mathbf V}
    \left(\ket{y}_Q\otimes I_{\Lin(\Hilb_E)}\right)
    \enspace .
    $$
    Thus, $x$ and $y$ refer only to the external output and input registers, respectively. It follows that
    \begin{align*}
    (T_{\mathbf V})_{x,y}
    =
    \sum_{r\in\{0,1\}^n}
    \braket{x}{r}
    \bigotimes_{i=1}^n B_{i,h_i(r)}^{r_i,y_i}
    =
    \bigotimes_{i=1}^n B_{i,h_i(x)}^{x_i,y_i}
    \enspace .
    \end{align*}

    We next compute the $Q$-blocks of $T_{\mathbf V}T_{\mathbf V}^\dagger$. For every $x,x'\in\{0,1\}^n$,
    \begin{equation}
    \begin{aligned}
    \left(T_{\mathbf V}T_{\mathbf V}^\dagger\right)_{x,x'}
    &=
    \sum_{y\in\{0,1\}^n}
    (T_{\mathbf V})_{x,y}
    \left((T_{\mathbf V})_{x',y}\right)^\dagger\\
    &=
    \sum_{y\in\{0,1\}^n}
    \bigotimes_{i=1}^n
    \left[
    B_{i,h_i(x)}^{x_i,y_i}
    \left(B_{i,h_i(x')}^{x_i',y_i}\right)^\dagger
    \right]\\
    &=
    \bigotimes_{i=1}^n
    \left[
    \sum_{v\in\{0,1\}}
    B_{i,h_i(x)}^{x_i,v}
    \left(B_{i,h_i(x')}^{x_i',v}\right)^\dagger
    \right]
    \enspace .
    \end{aligned}
    \label{eq:basis-full-valid}
    \end{equation}
    In the last equality, we used that, for fixed $x,x'$, the $i$-th factor depends on $y$ only through its $i$-th bit. Recall that for all $i, a$ the matrix $B_{i, a}$ is unitary and thus $B_{i,a}B_{i,a}^\dagger=I_{1 + |E_i|}$. For $u, u' \in \zo$, taking the $(u,u')$ part of the equality $B_{i,a}B_{i,a}^\dagger=I$ gives
    \begin{equation}
    \sum_{v\in\{0,1\}} B_{i,a}^{u,v}\left(B_{i,a}^{u',v}\right)^\dagger
    =\mathbf{1}[u=u']I_{\Lin(\Hilb_{E_i})}
    \enspace .\label{eq:basis-row}
    \end{equation}

    Back to analyzing $\left(T_{\mathbf V}T_{\mathbf V}^\dagger\right)_{x,x'}$. If $x=x'$, Equation \eqref{eq:basis-row} implies that every factor in Equation \eqref{eq:basis-full-valid} is $I_{\Lin(\Hilb_{E_i})}$, so the resulting matrix is $I_{\Lin(\Hilb_E)}$. Otherwise, Equation \eqref{equation:uso_combinatorial_condition} gives an index $i$ such that $x_i\neq x_i'$ and $h_i(x)=h_i(x')$. Write $a:=h_i(x)=h_i(x')$. At this coordinate, the factor is
    $$
    \sum_{v\in\{0,1\}}
    B_{i,a}^{x_i,v}
    \left(B_{i,a}^{x_i',v}\right)^\dagger
    =0
    \enspace ,
    $$
    where the equality follows from Equation \eqref{eq:basis-row}, because $x_i\neq x_i'$. Hence the entire tensor product is zero. We have proved that the diagonal matrices of $T_{\mathbf V}T_{\mathbf V}^\dagger$ are identities and all its off-diagonal matrices are zero. Consequently,
    $$
    T_{\mathbf V}T_{\mathbf V}^\dagger
    =
    I_{\Lin(\Hilb_Q\otimes\Hilb_E)}
    \enspace .
    $$
    Since the additional registers $E_i$ and the unitaries $V_i$ were arbitrary, this proves the equivalent condition in Definition \ref{definition:pure_process_matrix_generator}. Together with the unitarity of $C$, we conclude that $(n,C)$ is a pure PMG.

    \paragraph{The generated channel.}
    For the identity insertion, Lemma \ref{lemma:simplified_isometry_pm} identifies the generated unitary with $\tr_R\left( U_C \right)$. For every $y\in\{0,1\}^n$, its action is
    \begin{align*}
    \left (\tr_R \left( U_C \right) \right)\ket{y}
    &=\sum_{x\in\{0,1\}^n}(\bra{x}\otimes I_Q)U_C(\ket{x}\otimes\ket{y})\\
    &=\sum_{x\in\{0,1\}^n}\bra{x}U(x)\ket{y}\ket{x}
    =T\ket{y}
    \enspace ,
    \end{align*}
    for the matrix $T$ from the statement of Theorem \ref{theorem:decomposable_unitary_generation}.
    This proves Equation \eqref{eq:basis-map}. In particular, the generated operation sends $\sum_x\alpha_x\ket{b_x}$ to $\sum_x\alpha_x\ket{x}$, and the same equality holds with an identity operation on any untouched reference register. This concludes the proof.
\end{proof}

\paragraph{The special case of USO.}
Taking $U_{i,0}=I_1$ and $U_{i,1}=X$ in the preceding theorem gives the following result.

\begin{theorem} [Coherent inversion of a USO outmap] \label{thm:uso}
    There is a classical deterministic polynomial-time algorithm which, given a circuit for a USO outmap $s:\{0,1\}^n\to\{0,1\}^n$, outputs a pure PMG $(n,C_s)$ (w.r.t. Definition \ref{definition:pure_process_matrix_generator}) whose generated unitary is
    \begin{equation}
    T_s:=\sum_{x\in\{0,1\}^n}\ket{x}\bra{s(x)}
    \quad\text{and hence}\quad
    T_s\ket{y}=\ket{s^{-1}(y)}
    \enspace .\label{equation:uso_characterizationso-inverse}
    \end{equation}
    In particular, one query to the generated channel on $\ket{0^n}$ finds the unique sink with probability $1$.
\end{theorem}
\begin{proof}
    Compute $h(x)=x\oplus s(x)$ by an ordinary Boolean circuit. By Lemma \ref{lem:uso-characterization}, $h$ satisfies Equation \eqref{equation:uso_combinatorial_condition}. For the gates $U_{i,a}=X^a$, the circuit in Theorem \ref{theorem:decomposable_unitary_generation} is
    \begin{equation}
    U_{C_s} \left( \ket{x}_R\ket{y}_Q \right)
    =
    \ket{y\oplus x\oplus s(x)}_R \ket{x}_Q
    \enspace.
    \label{equation:uso_characterizationso-C}
    \end{equation}
    It is a pure PMG by that theorem. Moreover, $\ket{b_x}=\ket{x\oplus h(x)}=\ket{s(x)}$, which proves Equation \eqref{equation:uso_characterizationso-inverse}. The computational model gives access to the generated channel $T_s$. Applying it to $\ket{0^n}$ gives $\ket{s^{-1}\left(0^n\right)}$, the unique sink. 
\end{proof}

For general $U_{i,a}$, Theorem \ref{theorem:decomposable_unitary_generation} instead gives a quantum basis transformation. Thus the same validity argument supports both the permutation inverse and the quantum constructions below. The relation between process functions and orthogonal product bases has been studied in \cite{kunjwal2023trading}; here we specify a succinct circuit representation and prove the corresponding algorithm in $\BQPuctc$.

\begin{corollary}[Consequences of classical reductions] \label{corollary:catalogue}
    Every problem with a classical polynomial-time promise-preserving reduction to circuit-described Cube-USO has a polynomial-time algorithm using pure process computation.
\end{corollary}
\begin{proof}
    Compute the outmap circuit supplied by the reduction, use Theorem \ref{thm:uso} to find its sink, and decode the solution. For a polynomial-time Turing reduction, apply the same procedure to each USO instance it produces. On a promised source input, each such instance satisfies the USO promise, so our construction supplies a pure PMG for every call.
\end{proof}

Known reductions therefore give polynomial-time algorithms using pure process computation for rational P-matrix linear complementarity \cite{szabo2001unique}, its generalized and P--Lin-Bellman formulations, and Grid-USO with total-side-length encoding \cite{borzechowski2024two}. They also cover simple stochastic, discounted, mean-payoff, and parity games \cite{gartner2005simple,jurdzinski2008simple,fearnley2009linear}, quantitative stochastic parity games \cite{berthon2025direct}, and finite-state modal $\mu$-calculus model checking \cite{wilke2001alternating}. The game problems here have explicitly represented finite arenas and are turn-based, fully observed, two-player, and zero-sum, with binary rational numerical data where applicable. Model checking takes an explicitly represented finite transition system and a formula as input. We use these classical reductions directly; the process construction needed for each resulting USO instance is the one in Theorem \ref{thm:uso}.

\subsection{Connections to physical implementation} \label{subsection:uso_physical_implementation}
Theorem \ref{thm:uso} gives an efficient generating circuit for a pure process that solves $\USO$. Whether pure processes can be implemented physically with resources polynomial in the size of their generating circuits remains unknown. The structure of our construction gives concrete connections to existing circuit simulations and proposals for physical implementation, which we discuss below.

For our algorithm, a physical implementation would need to provide the generated channel, obtained when all local operations are identities. Reproducing the entire process under arbitrary local operations is a stronger requirement. This distinction already appears in \cite[Section VI]{araujo2017quantum}: they construct a family of pure processes that violate causal inequalities, yet whose connected channels admit efficient ordinary-circuit simulations when the local operations are given by circuit descriptions. Their generators use controlled bit flips followed by a SWAP, as in our construction.

Our more general construction also has a concrete connection to work on quantum control of causal order. Recall that Theorem \ref{theorem:decomposable_unitary_generation} generates the transformation
$$
T\ket{b_x}=\ket{x}
\quad\text{where}\quad
\ket{b_x}
=
\bigotimes_{i=1}^{n}
U_{i,h_i\left( x \right)}^\dagger\ket{x_i}
\enspace .
$$
Each basis vector is a product state, with only two possible single-qubit operations at each coordinate. For example, take $n=3$ and
$$
h\left( x \right)
=
\left(
\left( 1-x_2 \right)x_3,\,
x_1\left( 1-x_3 \right),\,
\left( 1-x_1 \right)x_2
\right)
\enspace .
$$
With $U_{i,0}=I$ and $U_{i,1}=H$, these vectors are exactly the SHIFT basis studied in \cite{kunjwal2023trading}. Steffinlongo and Dourdent \cite{steffinlongo2026simulating} show how a quantum switch can implement the corresponding projective measurement and propose a photonic realization. Thus, a specific case of our basis construction already has a switch-based implementation proposal. The proposal concerns measurement in this basis, rather than the coherent transformation $T$.

These connections suggest studying physical implementations of the generated transformation directly. For $\USO$, the target is particularly simple to state: take $U_{i,a}=X^a$, so that $\ket{b_x}=\ket{s\left( x \right)}$, and implement the map $\ket{s\left( x \right)}\mapsto\ket{x}$. The same two single-qubit gates are used for every instance; the instance enters through the efficiently computable function $h$. Whether this structure permits a physical implementation with resources polynomial in the description of the instance remains open. Existing simulations and switch-based proposals establish particular cases, but do not provide such a construction for arbitrary $\USO$ instances.

\vspace{7mm}

\paragraph{Acknowledgements.}
The main ideas in this paper were developed by the author. During its preparation, the author used ChatGPT (OpenAI) to help work through and check mathematical arguments, draft and revise proofs, explore the literature, and improve the presentation. The author takes full responsibility for all statements and proofs in this paper.

\ifllncs
\bibliographystyle{plain}
\else
\bibliographystyle{alpha}
\fi
\bibliography{bibliography,abbrev0}

\appendix

\section{Supplementary Proofs} \label{section:appedix}
We prove the pairwise characterization of $\USO$, show that the USO property is necessary for our particular generating circuit, and bound the effect of approximating the single-qubit gates in our more general construction.

\subsection{The USO characterization and validity of our construction} \label{app:uso-facts}
The characterization in Lemma \ref{lem:uso-characterization} compares pairs of vertices. We first show that, on every face of a USO, the outmap restricted to the free coordinates takes each possible value exactly once. The pairwise condition will then follow by considering the smallest face containing the two vertices.

\begin{proof}[Proof of Lemma \ref{lem:uso-characterization}]
    Suppose that $s$ is a USO outmap, and fix a face with free coordinates $J\subseteq[n]$. We want to show that
    $$
    x \longmapsto \left( s_i\left( x \right) \right)_{i\in J}
    $$
    is a bijection between the vertices of this face and $\zo^J$.

    First, fix $K\subseteq J$. Partition the face into the $2^{|J|-|K|}$ smaller faces whose free coordinates are precisely $K$. A vertex satisfies $s_i\left( x \right)=0$ for every $i\in K$ if and only if it is the sink of its smaller face. Since each smaller face has exactly one sink, there are exactly $2^{|J|-|K|}$ such vertices.

    Now fix $Z\subseteq J$. We count the vertices whose outmap is zero exactly on the coordinates in $Z$, among those in $J$. By inclusion--exclusion, this number is
    $$
    \sum_{K\subseteq J\setminus Z}
    \left( -1 \right)^{|K|}2^{|J|-|Z|-|K|}
    =
    2^{|J|-|Z|}\left( 1-\frac12 \right)^{|J|-|Z|}
    =1
    \enspace .
    $$
    Thus every possible outmap value occurs exactly once. For any distinct vertices $x,x'$, consider the smallest face containing both. Its free coordinates are exactly those where $x$ and $x'$ differ. Their outmap values on these coordinates must also differ, proving Equation \eqref{equation:uso_characterization}.

    Conversely, suppose Equation \eqref{equation:uso_characterization} holds. Apply it first to two vertices differing only at coordinate $i$. Their $i$-th outmap bits must differ, so $s$ assigns a consistent direction to their shared edge. Next, fix any face with free coordinates $J$. The same condition says that its restricted outmap is injective. Since the face and $\zo^J$ have the same size, this map is bijective. Exactly one vertex maps to the all-zero string on $J$, and this vertex is the unique sink of the face. Taking the entire cube also shows that $s$ is a permutation.

    Finally, whenever $x_i\neq x_i'$, we have
    $$
    x_i\oplus s_i\left( x \right)
    =x_i'\oplus s_i\left( x' \right)
    \quad \Longleftrightarrow \quad
    s_i\left( x \right)\neq s_i\left( x' \right)
    \enspace .
    $$
    Substituting $h\left( x \right)=x\oplus s\left( x \right)$ gives Equation \eqref{equation:uso_combinatorial_condition}.
\end{proof}

\paragraph{Why our generating circuit needs the USO property.}
Theorem \ref{thm:uso} shows that a USO gives a valid pure PMG. For the particular circuit used in that theorem, the converse also holds. To see why, we choose local channels that fix some feedback inputs and therefore restrict them to a face. On external input $\ket{0^n}$, the trace of the resulting output counts the sinks of that face. Validity requires this trace to be exactly one.

\begin{proposition}[The exact condition needed by our USO generator] \label{prop:uso-necessity}
    Let $s:\zo^n\to\zo^n$ be any Boolean function, and let $C_s$ be the unitary circuit in Equation \eqref{equation:uso_characterizationso-C}. Then $\left( n,C_s \right)$ is a PMG if and only if $s$ is a USO outmap.
\end{proposition}
\begin{proof}
    One direction is Theorem \ref{thm:uso}. For the other, suppose $\left( n,C_s \right)$ is a PMG. Fix a face by setting the coordinates in $I\subseteq[n]$ to a string $c\in\zo^I$, and write $J:=[n]\setminus I$ for its free coordinates. At each feedback pair $i\in I$, insert the channel that discards its input and prepares $\ket{c_i}$. At each feedback pair $j\in J$, insert the identity. Set the external input to $\ket{0^n}$.

    The reset channel at coordinate $i$ has Kraus operators
    $$
    \left\{ \ket{c_i}\bra{b_i}:b_i\in\zo \right\}
    \enspace .
    $$
    For each $b\in\zo^I$, let $K_b$ be the connected Kraus operator corresponding to these choices of $b_i$. Write $h\left( x \right)=x\oplus s\left( x \right)$. Expanding the operator partial trace gives
    $$
    K_b\ket{0^n}
    =
    \sum_{\substack{x_I=c;\ s_j\left( x \right)=0\ \text{for every }j\in J;\\
    h_i\left( x \right)=b_i\ \text{for every }i\in I}}
    \ket{x}
    \enspace .
    $$
    Indeed, the reset channels force $x_I=c$, and their Kraus operators select the feedback outputs $b_i=h_i\left( x \right)$. At a free coordinate $j$, the identity connection requires $h_j\left( x \right)=x_j$, which is equivalent to $s_j\left( x \right)=0$.

    Every vertex satisfying $x_I=c$ and $s_j\left( x \right)=0$ for all $j\in J$ appears for exactly one value of $b$, with amplitude one. Distinct vertices give orthogonal output states. Therefore,
    $$
    \sum_{b\in\zo^I}\norm{K_b\ket{0^n}}_2^2
    =
    \#\left\{ x\in\zo^n:x_I=c,\ s_j\left( x \right)=0\text{ for every }j\in J \right\}
    \enspace .
    $$
    Since the connected transformation is trace-preserving, this number must be one. On a one-dimensional face, this says that exactly one endpoint has outmap bit zero along its edge, so $s$ defines a consistent orientation. On any face, it then says that there is exactly one sink. Thus $s$ is a USO outmap.
\end{proof}

\paragraph{Approximating the single-qubit gates.}
Theorem \ref{theorem:decomposable_unitary_generation} applies to any choice of the single-qubit unitaries. Thus, replacing them by unitary approximations still gives a pure PMG. What changes is the generated unitary. The next lemma bounds this change in terms of the approximation error of the individual gates.

\begin{lemma}[Precision of the basis transformation] \label{lem:unitary-precision}
    Fix $h$ as in Theorem \ref{theorem:decomposable_unitary_generation}. Let $T,\widetilde T$ be the generated matrices for two collections of single-qubit unitaries $U_{i,a},\widetilde U_{i,a}$. If
    $$
    \norm{U_{i,a}-\widetilde U_{i,a}}_{\mathrm{op}}\leq\eta
    \quad\text{for every }i\in[n],\ a\in\zo,
    $$
    then
    \begin{equation}
    \norm{T-\widetilde T}_{\mathrm{op}}
    \leq\norm{T-\widetilde T}_{\mathrm F}
    \leq 2^{n/2}n\eta
    \enspace .
    \label{eq:precision}
    \end{equation}
    Here $\norm{M}_{\mathrm{op}}:=\sup_{\norm{v}_2=1}\norm{Mv}_2$ is the operator norm, and $\norm{M}_{\mathrm F}^2:=\sum_{x,y}\left| M_{x,y} \right|^2$ is the squared Frobenius norm. If both collections are implemented by unitary circuits, both generating circuits describe pure PMGs exactly.
\end{lemma}
\begin{proof}
    Fix a row $x$ of $T$. By Equation \eqref{eq:basis-map}, this row is
    $$
    \bigotimes_{i=1}^{n}\bra{x_i}U_{i,h_i\left( x \right)}
    \enspace .
    $$
    Replace these factors one at a time by the corresponding factors from $\widetilde T$. Each replacement changes the row by Euclidean norm at most $\eta$: the changed factor differs by norm at most $\eta$, and every other factor has norm one. By the triangle inequality, the two rows differ by norm at most $n\eta$. Summing the squared differences over the $2^n$ rows gives
    $$
    \norm{T-\widetilde T}_{\mathrm F}^2\leq 2^n n^2\eta^2
    \enspace .
    $$
    Taking square roots and using that the operator norm is at most the Frobenius norm proves Equation \eqref{eq:precision}. Both generating circuits describe pure PMGs by Theorem \ref{theorem:decomposable_unitary_generation}.

    The same bound controls the output trace distance, including when the input is entangled with a reference register. Tensoring with the identity on that register leaves the operator norm unchanged. For unit vectors $\ket{u},\ket{v}$, the trace distance between their density matrices is
    $$
    \sqrt{1-\left| \braket{u}{v} \right|^2}
    \leq\norm{\ket{u}-\ket{v}}_2
    \enspace .
    $$
    Applying this inequality to the two output states gives the claim for pure joint inputs. For a mixed joint input, first purify it and then discard only the additional purifying register. Trace distance cannot increase under this last operation, so the bound also holds for mixed inputs, with the original reference register included in the output.
\end{proof}

The approximation compares the unitary matrices with their specified phases. These phases matter here: a gate is selected according to the feedback input, so changing its phase can change the relative phases between different branches of the computation.

\end{document}